%% file: ltgs.tex
\documentclass[submission,copyright,creativecommons]{eptcs}
\providecommand{\event}{TERMGRAPH 2012} 

\usepackage{amsfonts,amsmath,amscd,amssymb,float}
\usepackage{latexsym,MnSymbol}
\usepackage[dutch,german,british,english]{babel}
\usepackage{comment}
\usepackage{xspace}
\usepackage{color,graphicx}
\usepackage[utf8]{inputenc}
\usepackage{url,breakurl,hyperref}
\usepackage{enumerate}
\usepackage{tikz}
\usepackage[normalem]{ulem}

\title{Term Graph Representations for Cyclic Lambda-Terms%
       \footnote{This work was started, and in part carried out, 
                 within the framework of the project NWO~project \emph{Realising Optimal Sharing (ROS)}, project number 612.000.935, 
                 under the direction of Vincent von Oostrom and Doaitse Swierstra.}}
\author{Clemens Grabmayer
\institute{Department of Philosophy\\
           Utrecht University\\
           The Netherlands}
\email{clemens@phil.uu.nl}
\and
Jan Rochel 
\institute{Department of Computing Sciences\\
           Utrecht University\\
           The Netherlands}
\email{jan@rochel.info}
}
\def\titlerunning{Term Graph Representations for Cyclic Lambda-Terms}
\def\authorrunning{C. Grabmayer \& J. Rochel}

\newcommand\itemizeprefs{
  \setlength\itemsep{-0.2ex}
  \vspace{-0.4ex}
}

\input{defs}

\newtheorem{theorem}{Theorem} 
\newtheorem{corollary}[theorem]{Corollary}
\newtheorem{lemma}[theorem]{Lemma}
\newtheorem{proposition}[theorem]{Proposition}

\newtheorem{definition}[theorem]{Definition}
\newtheorem{remark}[theorem]{Remark}
\newtheorem{example}[theorem]{Example}

\newcommand{\qedsym}{\inMath{\Box}}
\newcommand{\qed}{\hspace*{\fill}\qedsym}
\newenvironment{proof}{\paragraph{\normalfont{\emph{Proof}.}}}{\qed\medskip}
\newenvironment{proofidea}{\paragraph{\normalfont{\emph{Proof (Idea)}.}}}{\qed\medskip}

\renewcommand{\emph}[1]{{\em #1}}

\begin{document}
\maketitle

\begin{abstract}
We study various representations for cyclic \lambda-terms as higher-order or as first-order term graphs.
We focus on the relation between `$\lambda$-higher-order term graphs'
(\lambdahotg{s}), which are first-order term graphs endowed with a well-behaved scope function, 
and their representations 
as `\lambdatg{s}', which are plain first-order term graphs with scope-delimiter vertices 
that meet certain scoping requirements. 
Specifically we tackle the question: 
Which class of first-order term graphs
admits a faithful embedding of \lambdahotg{s} in the sense that
(i)~the homomorphism-based sharing-order on \lambdahotg{s} is preserved and reflected,
and 
(ii)~the image of the embedding corresponds closely to a natural class (of \lambdatg{s})
     that is closed under homomorphism?

We systematically examine whether a number of classes of \lambdatg{s} have this property,  
and we find a particular class of \lambdatg{s} that satisfies this criterion.
Term graphs of this class are built from application, abstraction, variable, and scope-delimiter vertices,
and have the characteristic feature that the latter two kinds of vertices have back-links to the corresponding abstraction.

This result puts a handle on the concept of subterm sharing for
higher-order term graphs, both theoretically and algorithmically:
We obtain an easily implementable method for obtaining 
the maximally shared form of \lambdahotg{s}.
Furthermore, we open up the possibility to pull back properties 
from first-order term graphs to \lambdahotg{s}, properties 
such as the complete lattice structure of bisimulation equivalence classes
with respect to the sharing order.
\end{abstract}

\section{Introduction}\label{sec:intro}
%
Cyclic lambda-terms typically represent infinite \lambda-terms. In this paper we study term graph
representations of cyclic \lambda-terms and their respective notions of
homomorphism, or functional bisimulation.

The context in which the results presented in this paper play a central role is
our research on subterm sharing as present in terms of languages such as the
\lambdacalculus\ with \stxtletrec\ \cite{peyt:jone:1987,ario:blom:1997}, 
with recursive definitions \cite{ario:klop:1994}, or languages with $\mu$\nb-recursion~\cite{ario:klop:1996},
and our interest in describing maximal sharing in such settings. 
Specifically we want to obtain concepts and methods as follows:
\begin{itemize}\itemizeprefs
\item an efficient test for term equivalence with respect to \alpha-renaming and unfolding;
\item a notion of `maximal subterm sharing' for terms in the respective language;
\item the efficient computation of the maximally shared form of a term;
\item a sharing (pre-)order on unfolding-equivalent terms.
\end{itemize}\vspace*{-1ex}
Now our approach is to split the work into a part that concerns properties
specific to concrete languages, and into a part that deals with aspects that
are common to most of the languages with constructs for expressing subterm sharing. 
To this end we set out to find classes of term graphs 
that facilitate faithful interpretations of terms in such languages as 
(higher-order, and eventually first-order) term graphs,
and that are `well-behaved' in the sense that maximally shared term graphs do always exist.
In this way the task can be divided into two parts:
an investigation of sharing for term graphs with higher-order features (the aim of this paper),
and a study of language-specific aspects of sharing (the aim of a further paper).

Here we study a variety of classes of term graphs for denoting cyclic \lambdaterms,
term graphs with higher-order features and their first-order `implementations'. 
All higher-order term graphs we consider are built from three kinds of vertices,
which symbolize applications, abstractions, and variable occurrences, respectively. 
They also carry features that describe notions of \emph{scope}, 
which are subject to certain conditions that guarantee the meaningfulness of the term graph
(that a \lambdaterm\ is denoted), and in some cases are crucial to define \emph{binding}. 
The first-order implementations do not have these additional features,
but they may contain scope-delimiter vertices. 

\enlargethispage{2ex}
In particular we study the following three kinds (of classes) of term graphs:
\renewcommand{\descriptionlabel}[1]%
      {\hspace{\labelsep}{\emph{#1}}}
\begin{description}\itemizeprefs
\item[$\lambda$-Higher-order-term-graphs {\normalfont (Section~\ref{sec:lambdahotgs})}]
  are extensions of first-order term graphs by adding a scope function that assigns a set of vertices, its scope, to every
  abstraction vertex. There are two variants,
  one with and one without an edge (a \emph{back-link}) from each variable occurrence to
  its corresponding abstraction vertex. The class
  with back-links is related to \emph{higher-order term graphs}
  as defined by Blom in \cite{blom:2001}, 
  and in fact is an adaptation of that concept for the purpose of representing \lambdaterms. 
\item[Abstraction-prefix based $\lambda$-higher-order-term-graphs {\normalfont (Section~\ref{sec:lambdaaphotgs})}]
  do not have a scope function but assign, to each vertex $\bvert$,
  an abstraction prefix consisting of a word of abstraction vertices 
  that includes those abstractions for which $\bvert$ is in their scope
  (it actually lists all abstractions for which $\bvert$ is in their `extended scope' \cite{grab:roch:2012}). 
  Abstraction prefixes are aggregations of scope information that is relevant for and locally available at individual vertices. 
\item[$\lambda$-Term-graphs with scope delimiters {\normalfont (Section~\ref{sec:ltgs})}]
  are plain first-order term graphs intended to represent higher-order term graphs of the two sorts above,
  and in this way stand for \lambdaterms. 
  Instead of relying upon additional features for describing scopes, 
  they use scope-delimiter vertices to signify the end of scopes. 
  Variable occurrences as well as scoping delimiters may
  or may not have back-links to their corresponding abstraction vertices.
\end{description}
%
%
Each of these classes induces a notion of homomorphism
(functional bisimulation) and bisimulation. 
Homomorphisms increase sharing in term graphs, 
and in this way induce a sharing order.
They preserve the unfolding semantics of term graphs%
  \footnote{While this is well-known for first-order term graphs, it can also be proved
   for the higher-order term graphs considered here.},
and therefore are able to preserve \lambdaterms\ that are denoted by term graphs in the unfolding semantics. 
Term graphs from the classes we consider always represent finite or infinite \lambdaterms,
and in this sense are not `meaningless'.
But this is not shown here. Instead, we lean on motivating examples, intuitions,
and the concept of higher-order term graph from \cite{blom:2001}.  

We 
establish a bijective correspondence between the former two classes, and a correspondence between the latter
two classes that is `almost bijective' 
(bijective up to sharing or unsharing of scope delimiter vertices).
All of these correspondences preserve and reflect the sharing order.
Furthermore, we systematically investigate which specific class of \lambdatg{s}
is closed under homomorphism and renders the mentioned correspondences possible.
We prove (in Section~\ref{sec:not:closed}) that this can only hold for a class 
in which both variable-occurrence and scope-delimiter vertices
have back-links to corresponding abstractions, 
and establish (in Section~\ref{sec:closed}) that the subclass containing only 
\lambdatg{s} with eager application of scope-closure satisfies these properties.
For this class the correspondences allow us:
\begin{itemize}\itemizeprefs
\item 
  to transfer properties known for first-order term graphs, such as the existence of a maximally shared form,
  from \lambdatg{s} to the corresponding classes of higher-order \lambdatg{s};
\item 
  to implement maximal sharing for higher-order \lambdatg{s} (with eager scope
  closure) via bisimulation collapse of the corresponding first-order
  \lambdatg{s} (see algorithm in Section~\ref{sec:conclusion}).
\end{itemize}



We stress that this paper in its present form is only a report about work in progress,
and, while a number of proofs are included, predominantly has the character of an extended abstract.

\section{Preliminaries}
  \label{sec:prelims}

By $\nats$ we denote the natural numbers including zero.
For words $w$ over an alphabet $A$ we denote the length of $w$ by $\length{w}$.
For a function $\safun \funin A \to B $ 
we denote by $\dom{\safun}$ the domain, and by $\image{\safun}$ the image of $\safun$;
and for $A_0\subseteq A$ we denote by $\srestrictfunto{\safun}{A_0}$ the restriction of $\safun$ to $A_0$. 

Let $\asig$ be a signature with arity function $\sarity \funin \asig \to \nats$.
A \emph{term graph over $\asig$} is a tuple $\tuple{\verts,\svlab,\svargs,\sroot}$ 
where $\verts$ is a set of \emph{vertices},
$\svlab \funin \verts \to \asig$ the \emph{(vertex) label function},
$\svargs \funin \verts \to \verts^*$ the \emph{argument function} 
  that maps every vertex $\avert$ to the word $\vargs{\avert}$ consisting of the $\arity{\vlab{\avert}}$ successor vertices of $\avert$
  (hence it holds $\length{\vargs{\avert}} = \arity{\vlab{\avert}}$),
$\sroot$, the \emph{root} is a vertex in $\verts$,
and where every vertex is reachable from the root (by a path that arises by repeatedly going from a vertex to one of its successors).
(Note this reachability condition, and mind the fact that term graphs may have infinitely many vertices.)
By a \emph{$\asig$\nb-term-graph} 
we mean a term graph over $\asig$.
And by $\classtgsover{\asig}$ we mean the class of all term graphs over $\asig$.

Let $\atg$ be a term graph over signature $\asig$. 
As useful notation for picking out any vertex or the $i$-th vertex 
from among the ordered successors of a vertex $\avert$ in $\atg$
we define the (not indexed) edge relation ${\stgsucc} \subseteq \verts\times\verts$,
and for each $i\in\nats$ the indexed edge relation \inMath{{\stgsucci{i}} \subseteq {\verts\times\verts}},
between vertices by stipulating that:
\begin{align*}
  \bvert \tgsucci{i} \bvertacc
    \;\;\funin\,&\Longleftrightarrow\;\;
  \existsst{\bverti{0},\ldots,\bverti{n}\in\verts}
           {\,\vargs{\bvert} = \bverti{0}\ldots\bverti{n}
                \;\logand\;
              \bvertacc = \bverti{i}}
  \\
  \bvert \tgsucc \bvertacc 
    \;\;\funin\,& \Longleftrightarrow\;\;
  \existsst{i\in\nats}{\, \bvert \stgsucci{i} \bvertacc} 
\end{align*}
holds for all $\bvert,\bvertacc\in\verts$.
We write 
$\bvert \tgsuccis{i}{f} \bvertacc$
if 
$\bvert \tgsucci{i} \bvertacc 
   \logand
 \vlab{\bvert} = f$
holds for $\bvert,\bvertacc\in\verts$, $i\in\nats$, $f\in\asig$,
to indicate the label at the source of an edge.
A \emph{path} in $\atg$ is a tuple of the form
$\tuple{\bverti{0},k_0,\bverti{1},k_1,\bverti{2},\ldots,\bverti{n-1},k_{n-1},\bverti{n}}$
where $\bverti{0},\ldots,\bverti{n}\in\verts$ and $n,k_0,k_1,\ldots,k_{n-1}\in\nats$
such that 
$\bverti{0} \tgsucci{k_0} \bverti{1} \tgsucci{k_1} \bverti{2} \mathrel{\cdots} \bverti{n-1} \tgsucci{k_{n-1}} \bverti{n}$
holds; paths will usually be denoted in the latter form, using indexed edge relations. 
An \emph{access path} of a vertex $\bvert$ of $\atg$ is
a path that starts at the root of $\atg$, ends in $\bvert$, and does not visit any vertex twice. 
Note that every vertex $\bvert$ has at least one access path: 
since every vertex in a term graph is reachable from the root, 
there is a path $\apath$ from $\sroot$ to $\bvert$;
then an access path of $\bvert$ can be obtained from $\apath$ 
by repeatedly cutting out \emph{cycles}, that is,
parts of the path between different visits to one and the same vertex. 

In the sequel, let $\atgi{1} = \tuple{\vertsi{1},\svlabi{1},\svargsi{1},\srooti{1}}$,
  $\atgi{2} = \tuple{\vertsi{2},\svlabi{2},\svargsi{2},\srooti{2}}$
be term graphs over signature $\asig$. 

A \emph{homomorphism}, also called a \emph{functional bisimulation}, 
from $\atgi{1}$ to $\atgi{2}$ 
is a morphism from the structure
$\tuple{\vertsi{1},\svlabi{1},\svargsi{1},\srooti{1}}$
to the structure
$\tuple{\vertsi{2},\svlabi{2},\svargsi{2},\srooti{2}}$,
that is, a function 
$ \sahom \funin \vertsi{1} \to \vertsi{2}$ 
such that, for all $\avert\in\vertsi{1}$ it holds:
\begin{equation}\label{eq:def:homom}
\left.\qquad
\begin{aligned}
    \vlabi{1}{\avert}
    & = \vlabi{2}{\ahom{\avert}} 
  & & & & & (\text{labels})  
  \\
  \funap{\sahomext}{\vargsi{1}{\avert}}
    & = \vargsi{2}{\ahom{\avert}}
  & & & & & (\text{arguments})
  \\
  \ahom{\srooti{1}} 
    & = \srooti{2} 
  & & & & & (\text{roots})
\end{aligned}
\qquad\right\}
\end{equation}
where $\sahomext$ is the homomorphic extension of $\sahom$ to words over $\vertsi{1}$, that is, to the function 
$\sahomext \funin \vertsi{1}^* \to \vertsi{2}^*$, 
$ \averti{1}\ldots\averti{n}   \mapsto \ahom{\averti{1}}\ldots\ahom{\averti{n}} $.  
In this case we write $\atgi{1} \funbisimi{\sahom} \atgi{2}$,
or $\atgi{2} \convfunbisimi{\sahom} \atgi{1}$.
And we write
$\atgi{1} \funbisim \atgi{2}$,
or for that matter $\atgi{2} \convfunbisim \atgi{1}$,
if there is a homomorphism (a functional bisimulation) from $\atgi{1}$ to $\atgi{2}$.

Let $\afunsym\in\asig$. 
We write $\atgi{1} \funbisims{\afunsym} \atgi{2}$ or $\atgi{2} \convfunbisims{\afunsym} \atgi{1}$
if there is a homomorphism $\sahom$ between $\atgi{1}$ and $\atgi{2}$ with the property that
for all $\bverti{1},\bverti{2}\in\vertsi{1}$ with $\bverti{1}\neq\bverti{2}$ it holds that
$\ahom{\bverti{1}} = \ahom{\bverti{2}} 
   \;\Rightarrow\;
 \vlabi{1}{\bverti{1}} = \vlabi{1}{\bverti{2}} = \afunsym$,
that is, if $\sahom$ only `shares', or `identifies', vertices when they have label $\afunsym$.    
If $\sahom$ is such a homomorphism, we also write 
$\atgi{1} \funbisimis{\sahom}{\afunsym} \atgi{2}$ or $\atgi{2} \convfunbisimis{\sahom}{\afunsym} \atgi{1}$.

A \emph{bisimulation} between  
$\atgi{1}$ and $\atgi{2}$ 
is a term graph 
$ \atg = \tuple{\abisim,\svlab,\svargs,\sroot}$ over $\asig$
with $\abisim \subseteq \vertsi{1}\times\vertsi{2}$ 
 and $\sroot = \pair{\srooti{1}}{\srooti{2}}$
such that
$ \atgi{1} \convfunbisimi{\sproji{1}} \atg \funbisimi{\sproji{2}} \atgi{2}$
where $\sproji{1}$ and $\sproji{2}$ are projection functions, defined, for $i\in\setexp{1,2}$,
by $ \sproji{i} \funin \vertsi{1}\times\vertsi{2} \to \vertsi{i} $,
$\pair{\averti{1}}{\averti{2}} \mapsto \averti{i}$.
In this case we write $\atgi{1} \bisimi{\abisim} \atgi{2}$. 
And we write
$\atgi{1} \bisim \atgi{2}$
if there is a bisimulation between $\atgi{1}$ and $\atgi{2}$.

Alternatively, bisimulations for term graphs can be defined directly 
as relations on the vertex sets, 
obtaining the same notion of bisimilarity.
In this formulation,
a bisimulation between $\atgi{1}$ and $\atgi{2}$ is
a relation $\abisim \subseteq \vertsi{1}\times\vertsi{2}$ such that
the following conditions hold, for all $\pair{\avert}{\avertacc}\in\abisim$:
\begin{center}
$  
\begin{array}{cccr}
  \pair{\srooti{1}}{\srooti{2}} \in \abisim
    & \hspace*{9ex} & (\text{roots})
  \\[0.75ex]
  \vlabi{1}{\avert} = \vlabi{2}{\avertacc} 
    & & (\text{labels})  
  \\[0.75ex]
  \pair{\vargsi{1}{\avert}}{\vargsi{2}{\avertacc}} \in \abisim^*
    & & (\text{arguments})
\end{array}
$
\end{center}
where $\abisim^* \defdby \descsetexp{\pair{\averti{1}\cdots\averti{k}}{\avertacci{1}\cdots\avertacci{k}}}
                                    {\averti{1},\ldots,\averti{k}\in\vertsi{1},\, 
                                     \avertacci{1},\ldots,\avertacci{k}\in\vertsi{2}
                                     \text{ for $k\in\nats$ such that}
                                     \pair{\averti{i}}{\avertacci{i}}\in\abisim
                                     \text{ for all $1\le i\le k$}} 
       \,$.

%

Bisimulation is an equivalence relation on the class $\classtgsover{\asig}$ of term graphs over a signature $\asig$.
The homomorphism (functional bisimulation) relation $\sfunbisim$
is a pre-order on term graphs over a given signature $\asig$, and it induces a partial order on
isomorphism equivalence classes of term graphs over $\asig$.
We will refer to $\sfunbisim$ as the \emph{sharing pre-order}, and will speak of it as \emph{sharing order},
dropping the `pre'. 
The bisimulation equivalence class 
$\eqcl{\eqcl{\atg}{\siso}}{\sbisimsubscript} \defdby \descsetexp{\eqcl{\atgacc}{\siso}}{\atgacc \bisim \atg}$
of the isomorphism equivalence class $\eqcl{\atg}{\siso}$ of a term graph $\atg$ 
is ordered by homomorphism $\sfunbisim$ 
such that
$\pair{\eqcl{\eqcl{\atg}{\siso}}{\sbisimsubscript}}{\sfunbisim}$ is a complete lattice \cite{ario:klop:1996, terese:2003}.%
Note that, different from e.g.\ \cite{terese:2003}, 
we use the order relation $\sfunbisim$ in the same direction as $\le\,$:
if $\atgi{1} \funbisim \atgi{2}$, then $\atgi{2}$ is greater or equal to $\atgi{1}$ 
in the ordering $\sfunbisim$ (indicating that sharing is typically increased from $\atgi{1}$ to $\atgi{2}$). 

Let $\aclass \subseteq\classtgsover{\asig}$ be a subclass of the term graphs over $\asig$, for a signature $\asig$.
We say that $\aclass$ is \emph{closed under homorphism} (\emph{closed under bisimulation})
if $\atg \funbisim \atgacc$  (resp.\ $\atg \bisim \atgacc$)
   for $\atg,\atgacc\in\classtgsover{\asig}$ with $\atg\in\aclass$
   implies $\atgacc\in\aclass$.
Note these concepts are invariant under considering other signatures $\asigacc$
with $\aclass\subseteq\classtgsover{\asigacc}$.

\section{$\lambda$-higher-order-Term-Graphs}
  \label{sec:lambdahotgs}

By $\siglambda$ we designate the signature
$\setexp{ \sslapp, \sslabs }$ 
with $\arity{\sslapp} = 2$, and
     $\arity{\sslabs} = 1$. 
By $\siglambdai{i}$, for $i\in\setexp{0,1}$, we denote the extension
$\siglambda \cup \setexp{ \snlvar }$ of $\siglambda$
where $\arity{\snlvar} = i$.     
The classes of term graphs over $\siglambdai{0}$ and $\siglambdai{1}$
are denoted by $\classtgssiglambdai{0}$ and $\classtgssiglambdai{1}$, respectively. 
   
Let $\atg = \tuple{\verts,\svlab,\svargs,\sroot}$ be a term graph over a signature
extending $\siglambda$ or $\siglambdai{i}$, for $i\in\setexp{0,1}$.   
By $\vertsof{\sslabs}$ we designate the set of 
\emph{abstraction vertices} of $\atg$,
that is, the subset of $\verts$ consisting of all vertices with label $\sslabs$;
more formally,
$\vertsof{\sslabs} \defdby \descsetexp{\avert\in\verts}{\vlab{\avert} = \sslabs}$.
Analogously, the sets $\vertsof{\sslapp}$ and $\vertsof{\snlvar}$ 
of \emph{application vertices} and \emph{variable vertices} of $\atg$
are defined as the sets consisting of all vertices in $\verts$ with label $\sslapp$ or label $\snlvar$, respectively.  
Whether the variable vertices have an outgoing edge depends on the value of $i$.
The intention is to consider two variants of term graphs, one with and one
without variable back-links to their corresponding abstraction.

A `$\lambda$\nb-higher-order-term-graph'
consists of a $\siglambdai{i}$\nb-term-graph together with a scope function 
that maps abstraction vertices to their scopes (`extended scopes' in \cite{grab:roch:2012}), which are subsets of the set of vertices.

\begin{definition}[\lambdahotg]\label{def:lambdahotg}\normalfont
  Let $i\in\setexp{0,1}$. 
  A \emph{\lambdahotg}
  (short for \emph{$\lambda$\nb-higher-order-term-graph}) 
  \emph{over $\siglambdai{i}$},
  is a five-tuple $\alhotg = \tuple{\verts,\svlab,\svargs,\sroot,\sScope}$
  where $\atgi{\alhotg} = \tuple{\verts,\svlab,\svargs,\sroot}$ is a $\siglambdai{i}$\nb-term-graph,
  called the term graph \emph{underlying} $\alhotg$,
  and $\sScope \funin \vertsof{\sslabs} \to \powersetof{\verts}$
  is the \emph{scope function} of $\alhotg$ 
  (which maps an abstraction vertex $\avert$ to a set of vertices called its \emph{scope})
  that together with $\atgi{\alhotg}$ fulfills
  the following conditions:
  For all $k\in\setexp{0,1}$,
  all vertices $\bvert,\bverti{0},\bverti{1}\in\verts$,
  and all abstraction vertices $\avert,\averti{0},\averti{1}\in\vertsof{\sslabs}$
  it holds:
\begin{align*}
       \;\; & \Rightarrow\;\;
  \sroot\notin\Scopemin{\avert}
  & & \text{(root)} 
  \displaybreak[0]\\
       \;\; & \Rightarrow\;\;
  \avert \in \Scope{\avert} 
  & & \text{(self)}
  \displaybreak[0]\\ 
  \averti{1}\in\Scopemin{\averti{0}}
      \;\; & \Rightarrow\;\;
  \Scope{\averti{1}} \subseteq \Scopemin{\averti{0}} 
  & & \text{(nest)}
  \displaybreak[0]\\
  \bvert \tgsucci{k} \bverti{k}
    \;\logand\;
  \bverti{k}\in\Scopemin{\avert}
       \;\; & \Rightarrow\;\;  
  \bvert\in\Scope{\avert} 
  & & \text{(closed)}
  \displaybreak[0]\\
  \bvert\in\vertsof{\snlvar}
      \;\; & \Rightarrow\;\;
  \existsst{\averti{0}\in\vertsof{\sslabs}}{\bvert\in\Scopemin{\averti{0}}}  
  & & \text{(scope)$_0$} 
  \displaybreak[0]\\[0.5ex]
  \bvert\in\vertsof{\snlvar}
    \;\logand\;
  \bvert \tgsucci{0} \bverti{0}
      \;\; & \Rightarrow\;\;
  \left\{\hspace*{1pt}   
  \begin{aligned}[c]    
    & \bverti{0}\in\vertsof{\sslabs}
         \;\logand\;
    \\[-0.25ex]
    & \;\logand\; 
      \forall \avert\in\vertsof{\sslabs}.
    \\[-0.5ex]
    & \phantom{\;\logand\;} \hspace*{3ex}  
        (\bvert\in\Scope{\avert} \Leftrightarrow \bverti{0}\in\Scope{\avert})
  \end{aligned}  
  \right.
  & & \text{(scope)$_1$} 
\end{align*}%
where $\Scopemin{\avert} \defdby \Scope{\avert}\setminus\setexp{\avert}$.  
Note that if $i=0$, then (scope)$_1$ is trivially true and hence superfluous,
and if $i=1$, then (scope)$_0$ is redundant, because it follows from (scope)$_1$ in this case.  
For $\bvert\in\verts$ and $\avert\in\vertsof{\sslabs}$
we say that $\avert$ is a \emph{binder for} $\bvert$ if $\bvert\in\Scope{\avert}$,
and we designate by
$\binders{\bvert}$  
the set of binders of $\bvert$.

The classes of \lambdahotg{s} over $\siglambdai{0}$ and $\siglambdai{1}$ 
will be denoted by $\classlhotgsi{0}$ and $\classlhotgsi{1}$.
\end{definition}

See Fig.~\ref{fig:lambdahotg} for two different \lambdahotg{s} over $\siglambdai{i}$
both of which represent the same term in the \lambdacalculus\ with $\stxtletrec$, namely
$\letrecin{\arecvar = \labs{\avar}{\lapp{(\labs{\bvar}{\lapp{\bvar}{(\lapp{\avar}{\brecvar})}})}{(\labs{\cvar}{\lapp{\brecvar}{\arecvar}})}},\; 
           \brecvar = \labs{\dvar}{\dvar}}
          {\arecvar}$.

\begin{figure}[t]
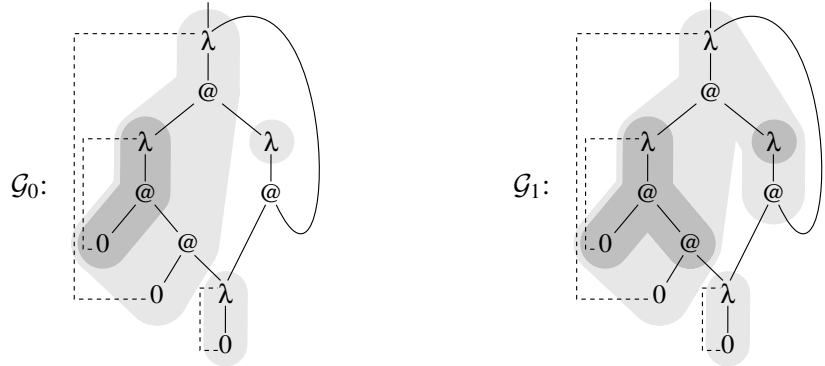

\graphs{
  \graph{\alhotgi{0}}{running_scopes_eager}
  \graph{\alhotgi{1}}{running_scopes_non_eager}
}
\vspace*{-2ex}
\caption{\label{fig:lambdahotg}
$\alhotgi{0}$ and $\alhotgi{1}$ are \lambdahotg{s} in $\classlhotgsi{i}$ whereby the
dotted back-link edges are present for $i=1$, but absent for
$i=0$. The underlying term graphs of $\alhotgi{0}$ and $\alhotgi{1}$ are
identical but their scope functions (signified by the shaded areas) differ. 
While in $\alhotgi{0}$ scopes are chosen as small as possible, which we refer to as `eager scope closure', 
in $\alhotgi{1}$ some scopes are closed only later in the graph.
}
\end{figure}

The following lemma states some basic properties of the scope function in \lambdahotg{s}.
Most importantly, scopes in \lambdahotg{s} are properly nested, in analogy with
scopes in finite \lambdaterms. 

\begin{lemma}\label{lem:lhotg}
  Let $i\in\setexp{0,1}$, and let $\alhotg = \tuple{\verts,\svlab,\svargs,\sroot,\sScope}$ be a \lambdahotg\ over $\siglambdai{i}$. 
  Then the following statements hold for all $\bvert\in\verts$ and $\avert,\averti{1},\averti{2}\in\vertsof{\sslabs}$:
  \begin{enumerate}[(i)]
    \item{}\label{lem:lhotg:item:i}
      If $\bvert\in\Scope{\avert}$, then $\avert$ is visited on every access path of $\bvert$,
      and all vertices on access paths of $\bvert$ after $\avert$ are in $\Scopemin{\avert}$.
      Hence (since $\atgi{\alhotg}$ is a term graph, every vertex has an access path) $\binders{\bvert}$ is finite.
    \item{}\label{lem:lhotg:item:ii}   
      If $\Scope{\averti{1}}\cap\Scope{\averti{2}} \neq \emptyset$ for $\averti{1} \neq \averti{2}\,$,
      then $\Scope{\averti{1}} \subseteq \Scopemin{\averti{2}}$ or $\Scope{\averti{2}} \subseteq \Scopemin{\averti{1}}$.
      As a consequence,
      if\/~$\Scope{\averti{1}} \cap \Scope{\averti{2}} \neq \emptyset$,
      then 
      $\Scope{\averti{1}} \subsetneqq \Scope{\averti{2}}$
      or $\Scope{\averti{1}} = \Scope{\averti{2}}$
      or $\Scope{\averti{2}} \subsetneqq \Scope{\averti{1}}$.
    \item{}\label{lem:lhotg:item:iii}
      If\/ $\binders{\bvert} \neq \emptyset$, 
      then $\binders{\bvert} = \setexp{\averti{0},\ldots,\averti{n}}$ 
      for $\averti{0},\ldots,\averti{n}\in\vertsof{\sslabs}$
      and $\Scope{\averti{n}} \subsetneqq \Scope{\averti{n-1}} \ldots \subsetneqq \Scope{\averti{0}}$.
  \end{enumerate}
\end{lemma}

\begin{proof}
  Let $i\in\setexp{0,1}$, and let $\alhotg = \tuple{\verts,\svlab,\svargs,\sroot,\sScope}$ 
  be a \lambdahotg\ over $\siglambdai{i}$. 
  
  For showing (\ref{lem:lhotg:item:i}),
  let $\bvert\in\verts$ and $\avert\in\vertsof{\sslabs}$ 
  be such that $\bvert\in\Scope{\avert}$.
  Suppose that 
  $\apath \funin \sroot = \bverti{0} \tgsucci{k_0} \bverti{1} \tgsucci{k_1} \bverti{2} \cdots \tgsucci{k_{n-1}} \bverti{n} = \bvert$
  is an access path of $\bvert$. 
  If $\bvert = \avert$, then nothing remains to be shown.
  Otherwise $\bverti{n} = \bvert \in\Scopemin{\avert}$, and, if $n>0$, then by (closed) it follows that $\bverti{n-1}\in\Scope{\avert}$.
  This argument can be repeated to find subsequently smaller $i$
  with $\bverti{i}\in\Scope{\avert}$ and $\bverti{i+1},\ldots,\bverti{n}\in\Scopemin{\avert}$. 
  We can proceed as long as $\bverti{i}\in\Scopemin{\avert}$.
  But since, due to (root), $\bverti{0} = \sroot\notin\Scopemin{\avert}$,  
  eventually we must encounter an $i_0$ such that
  such that $\bverti{i_0+1},\ldots,\bverti{n}\in\Scopemin{\avert}$
  and $\bverti{i_0}\in\Scope{\avert}\setminus\Scopemin{\avert}$.
  This implies $\bverti{i_0} = \avert$, showing that $\avert$ is visited on $\apath$.
  
  For showing (\ref{lem:lhotg:item:ii}),
  let $\bvert\in\verts$ and $\averti{1},\averti{2}\in\vertsof{\sslabs}$,
  $\averti{1} \neq \averti{2}$ 
  be such that
  $\bvert\in\Scope{\averti{1}}\cap\Scope{\averti{2}}$.
  Let $\apath$ be an access path of $\bvert$. 
  Then it follows by (\ref{lem:lhotg:item:i}) 
  that both $\averti{1}$ and $\averti{2}$ are visited on $\apath$,
  and that, depending on whether $\averti{1}$ or $\averti{2}$ is visited first on $\apath$,
  either $\averti{2}\in\Scopemin{\averti{1}}$ or $\averti{1}\in\Scopemin{\averti{2}}$. 
  Then due to (nest) it follows that either 
  $\Scope{\averti{2}}\subseteq\Scopemin{\averti{1}}$ holds or $\Scope{\averti{1}}\subseteq\Scopemin{\averti{2}}$.
  
  Finally, statement (\ref{lem:lhotg:item:iii}) is an easy consequence of statement (\ref{lem:lhotg:item:ii}).
\end{proof}

\begin{remark}\normalfont
  The notion of \lambdahotg\ is an adaptation of the notion of 
  `higher-order term graph' by Blom \cite[Def.\ 3.2.2]{blom:2001}
  for the purpose of representing finite or infinite \lambdaterms\ 
  or cyclic \lambdaterms, that is, terms in the \lambdacalculus\ with $\stxtletrec$.
  In particular, 
  \lambdahotg{s} over $\siglambdai{1}$ correspond closely to
  higher-order term graphs over signature~$\siglambda$. 
  But they differ in the following respects:
  \renewcommand{\descriptionlabel}[1]%
      {\hspace{\labelsep}{\emph{#1{:}}}}
\small
  \begin{description}\itemizeprefs
    \item[Abstractions]
      Higher-order term graphs in \cite{blom:2001} 
      are graph representations of finite or infinite terms 
      in Combinatory Reduction Systems (\CRS{s}).
      They typically contain abstraction vertices with label~$\Box$ 
      that represent \CRS\nb-ab\-strac\-tions.
      In contrast,
      \lambdahotg{s} have abstraction vertices with label~$\sslabs$ that denote
      $\lambda$\nb-abstractions.  
    \item[Signature] 
      Whereas higher-order term graphs in \cite{blom:2001}
      are based on an arbitrary \CRS\nb-signature,
      \lambdahotg{s} over $\siglambdai{1}$ 
      only contain 
      the application symbol $\sslapp$ and the vari\-able-occur\-rence symbol $\snlvar$
      in addition to the abstraction symbol $\sslabs$. 
    \item[Variable back-links and variable occurrence vertices]  
      In the formalization of higher-order term graphs in \cite{blom:2001}
      there are no explicit vertices that represent variable occurrences.
      Instead, variable occurrences are represented by back-link edges to abstraction vertices.
      Actually, in the formalization chosen in \cite[Def.\hspace*{2pt}3.2.1]{blom:2001}, 
      a back-link edge does not directly target the abstraction vertex $\avert$ it refers to, 
      but ends at a special variant vertex $\bar{\avert}$ of $\avert$. 
      (Every such variant abstraction vertex $\bar{\avert}$ could be looked upon as a variable vertex
       that is shared by all edges that represent occurrences of the variable bound by the abstraction vertex $\avert$.) 
      
      In \lambdahotg{s} over $\siglambdai{1}$
      a variable occurrence is represented by a variable-occurrence vertex
      that as outgoing edge has a back-link to the abstraction vertex that binds the occurrence. 
    \item[conditions on the scope function]
      While the conditions (root), (self), (nest), and (closed) on the scope function in higher-order term graphs
      in \cite[Def.\hspace*{2pt}3.2.2]{blom:2001} correspond directly to the respective conditions in Def.~\ref{def:lambdahotg},
      the difference between the condition (scope) there and (scope)$_1$ in Def.~\ref{def:lambdahotg}
      reflects the difference described in the previous item.
    \item[free variables]      
      Whereas the higher-order term graphs in \cite{blom:2001} cater for the
      presence of free variables, free variables have been excluded from the basic format of \lambdahotg{s}. 
  \end{description}
\end{remark}

\begin{definition}[homomorphism, bisimulation]\label{def:homom:lambdahotg}\normalfont
  Let $i\in\setexp{0,1}$. 
  Let $\alhotgi{1}$ and $\alhotgi{2}$ be \lambdahotg{s} over $\siglambdai{i}$ with
  $\alhotgi{k} = \tuple{\vertsi{k},\svlabi{k},\svargsi{k},\srooti{k},\sScopei{k}}$
  for $k\in\setexp{1,2}$. 

  A \emph{homomorphism}, also called a \emph{functional bisimulation}, 
  from $\alhotgi{1}$ to $\alhotgi{2}$ 
  is a morphism from the structure
  $\tuple{\vertsi{1},\svlabi{1},\svargsi{1},\srooti{1},\sScopei{1}}$
  to the structure
  $\tuple{\vertsi{2},\svlabi{2},\svargsi{2},\srooti{2},\sScopei{2}}$,
  that is, a function 
  $ \sahom \funin \vertsi{1} \to \vertsi{2}$ 
  such that, for all $\avert\in\vertsi{1}$ 
  the conditions (labels), (arguments), and (roots) in 
  in \eqref{eq:def:homom} are satisfied, and additionally,
  for all $\avert\in\vertsiof{1}{\sslabs}$:
\begin{equation}\label{eq:def:homom:lambdahotg}
\begin{aligned}
  \funap{\sahomextext}{\Scopei{1}{\avert}}
    & = \Scopei{2}{\ahom{\avert}}
  & & & & & (\text{scope functions})
\end{aligned}
\end{equation}
where $\sahomextext$ is the homomorphic extension of $\sahom$ to sets over $\vertsi{1}$, that is, to the function 
$\bar{\sahom} \funin \powersetof{\vertsi{1}} \to \powersetof{\vertsi{2}}$, 
$ A \mapsto \descsetexp{\ahom{a}}{a\in A} $.  
If there exists a homomorphism (a functional bisimulation) $\sahom$ from $\alhotgi{1}$ to $\alhotgi{2}$,
then we write $\alhotgi{1} \funbisimi{\sahom} \alhotgi{2}$
or $\alhotgi{2} \convfunbisimi{\sahom} \alhotgi{1}$,
or, dropping $\sahom$ as subscript, 
$\alhotgi{1} \funbisim \alhotgi{2}$ or $\alhotgi{2} \convfunbisim \alhotgi{1}$.

A \emph{bisimulation} between  
$\alhotgi{1}$ and $\alhotgi{2}$ 
is a term graph 
$ \alhotg = \tuple{\abisim,\svlab,\svargs,\sroot,\sScope}$ over $\asig$
with $\abisim \subseteq \vertsi{1}\times\vertsi{2}$ 
 and $\sroot = \pair{\srooti{1}}{\srooti{2}}$
such that
$ \alhotgi{1} \convfunbisimi{\sproji{1}} \alhotg \funbisimi{\sproji{2}} \alhotgi{2}$
where $\sproji{1}$ and $\sproji{2}$ are projection functions, defined, for $i\in\setexp{1,2}$,
by $ \sproji{i} \funin \vertsi{1}\times\vertsi{2} \to \vertsi{i} $,
$\pair{\averti{1}}{\averti{2}} \mapsto \averti{i}$.
If there exists a bisimulation $\abisim$ between $\alhotgi{1}$ and $\alhotgi{2}$,
then we write $\alhotgi{1} \bisimi{\abisim} \alhotgi{2}$, 
or just $\alhotgi{1} \bisim \alhotgi{2}$. 
\end{definition}

\section{Abstraction-prefix based $\lambda$-h.o.-term-graphs}
  \label{sec:lambdaaphotgs}

By an `abstraction-prefix based $\lambda$\nb-higher-order-term-graph' we will
mean a term-graph over $\siglambdai{i}$ for $i\in\setexp{0,1}$ that is endowed
with a correct abstraction prefix function that maps abstraction vertices
$\avert$ to words of vertices that represent the sequence of abstractions that
have $\avert$ in their scope. 
The conceptual difference between the abstraction-prefix function and the scope function
is that the former makes the most essential scoping information locally
available. It explicitly states all `extended scopes' (induced by the transitive closure of the in-scope relation, see \cite{grab:roch:2012})
in which a node resides in the order of their nesting. This approach leads to simpler correctness conditions.

\begin{definition}[correct abstraction-prefix function for $\siglambdai{i}$\nb-term-graphs]%
    \label{def:abspre:function:siglambdai}\normalfont
  Let $ \atg = \tuple{\verts,\svlab,\svargs,\sroot}$ be, for an $i\in\setexp{0,1}$,
  a $\siglambdai{i}$\nb-term-graph.
  
  A function $\sabspre \funin \verts \to \verts^*$
  from vertices of $\atg$ to words of vertices 
  is called an \emph{abstraction-prefix function} for $\atg$.
  Such a function is called \emph{correct} if
  for all $\bvert,\bverti{0},\bverti{1}\in\verts$ and $k\in\setexp{0,1}$:
  \begin{align*}
      \;\; & \Rightarrow \;\;
    \abspre{\sroot} = \emptyword  
    &
    (\text{root})
    \\
    \bvert\in\vertsof{\sslabs} 
      \;\logand\;
    \bvert \tgsucci{0} \bverti{0}
      \;\; & \Rightarrow \;\;
    \abspre{\bverti{0}} \prele \abspre{\bvert} \bvert
    & 
    (\sslabs)
    \\
    \bvert\in\vertsof{\sslapp}
      \;\logand\;
    \bvert \tgsucci{k} \bverti{k}
      \;\; & \Rightarrow \;\;
    \abspre{\bverti{k}} \prele \abspre{\bvert} 
    & 
    (\sslapp)
    \\
    \bvert\in\vertsof{\snlvar}
      \;\; & \Rightarrow \;\;
    \abspre{\bvert} \neq \emptyword  
    &
    (\snlvar)_{0}
    \\
    \bvert\in\vertsof{\snlvar}
      \;\logand\;
    \bvert \tgsucci{0} \bverti{0}
      \;\; & \Rightarrow \;\;
    \bverti{0}\in\vertsof{\sslabs}
      \;\logand\;
    \abspre{\bverti{0}}\bverti{0} = \abspre{\bvert}    
    & 
    (\snlvar)_{1}  
  \end{align*}
  Note that analogously as in Def.~\ref{def:lambdahotg},
  if $i=0$, then $(\snlvar)_1$ is trivially true and hence superfluous,
  and if $i=1$, then $(\snlvar)_0$ is redundant, because it follows from $(\snlvar)_1$ in this case. 

  We say that $\atg$ \emph{admits} a correct abstraction-prefix function if
  such a function exists for $\atg$.    
\end{definition}  

\begin{definition}[\lambdaaphotg]\label{def:aplambdahotg}\normalfont
  Let $i\in\setexp{0,1}$.
  A \emph{\lambdaaphotg} 
  (short for \emph{ab\-strac\-tion-pre\-fix based $\lambda$\nb-higher-order-term-graph})
  over signature $\siglambdai{i}$
  is a five-tuple $\alhotg = \tuple{\verts,\svlab,\svargs,\sroot,\sabspre}$
  where $\atgi{\alhotg} = \tuple{\verts,\svlab,\svargs,\sroot}$ is a $\siglambdai{i}$\nb-term-graph,
  called the term graph \emph{underlying} $\alhotg$,
  and $\sabspre$ is a correct \absprefix\ function for $\atgi{\alhotg}$.
  The classes of \lambdaaphotg{s} over $\siglambdai{i}$ 
  will be denoted by $\classlaphotgsi{i}$.
\end{definition}

See Fig.~\ref{fig:lambdaaphotg} for two examples, which correspond, as we will see, to the 
\lambdahotg{s} in Fig.~\ref{fig:lambdahotg}.

\begin{figure}[t]
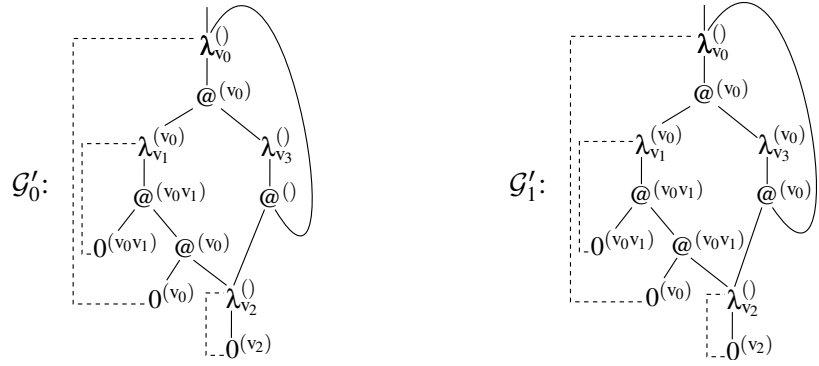

\graphs{
  \graph{\alaphotgi{0}'}{running_prefixed_ho_tg_eager}
  \graph{\alaphotgi{1}'}{running_prefixed_ho_tg_non_eager}
}
\vspace*{-2ex}
\caption{\label{fig:lambdaaphotg}
The \protect\lambdaaphotg{s} corresponding to the \protect\lambdahotg{s} in
Fig~\ref{fig:lambdahotg}. The subscripts of abstraction vertices indicate their names.
The super-scripts of vertices indicate their \absprefix{es}. A precise
formulation of this correspondence is given in
Example~\ref{ex:corr:lhotgs:laphotgs}.
}
\end{figure}

The following lemma states some basic properties of the scope function in \lambdaaphotg{s}.

\begin{lemma}\label{lem:laphotgs}
  Let $i\in\setexp{0,1}$ and let $\alaphotg = \tuple{\verts,\svlab,\svargs,\sroot,\sabspre}$ be a \lambdaaphotg\ over $\siglambdai{i}$. 
  Then the following statements hold:
  \begin{enumerate}[(i)]\itemizeprefs
    \item{}\label{lem:laphotg:item:i}
      Suppose that, for some $\avert,\bvert\in\verts$, $\avert$ occurs in $\abspre{\bvert}$.
      Then $\avert\in\vertsof{\sslabs}$, occurs in $\abspre{\bvert}$ only once,  
      and every access path of $\bvert$ passes through $\avert$, but does not end there, and thus $\bvert\neq\avert$.
      Furthermore it holds: $\stringcon{\abspre{\avert}}{\avert} \prele \abspre{\bvert}$.
      In particular, if $\abspre{\bvert} = \stringcon{\apre}{\avert}$, then $\abspre{\avert} = \apre$.
    \item{}\label{lem:laphotg:item:ii}
      Vertices in abstraction prefixes are abstraction vertices, and hence
      $\sabspre$ is of the form $\sabspre \funin \verts\to(\vertsof{\sslabs})^*$.  
    \item{}\label{lem:laphotg:item:iii}
      For all $\avert\in\vertsof{\sslabs}$ it holds: $\avert\notin\abspre{\avert}$. 
    \item{}\label{lem:laphotg:item:iv}
      While access paths might end in vertices in $\vertsof{\snlvar}$,
      they only pass through vertices in $\vertsof{\sslabs} \cup \vertsof{\sslapp}$.
  \end{enumerate}
\end{lemma}

\begin{proof}
  Let $i\in\setexp{0,1}$ and let $\alaphotg = \tuple{\verts,\svlab,\svargs,\sroot,\sabspre}$ be a \lambdaaphotg\ over $\siglambdai{i}$. 
  
  For showing (\ref{lem:laphotg:item:i}), let $\avert,\bvert\in\verts$ be such that $\avert$ occurs in $\abspre{\bvert}$. 
  Suppose further that $\apath$
  is an access path of $\bvert$. 
  Note that when walking through $\apath$ the abstraction prefix starts out empty (due to (root)),
  and is expanded only in steps from vertices $\avertacc\in\vertsof{\sslabs}$ 
  (due to ($\sslabs$), ($\sslapp$), and ($\snlvar$)$_1$)
  in which just $\avertacc$ is added to the prefix on the right (due to ($\sslabs$)).
  Since $\avert$ occurs in $\abspre{\bvert}$, it follows that $\avert\in\vertsof{\sslabs}$, 
  that $\avert$ must be visited on $\apath$, 
  and that $\apath$ continues after the visit to $\avert$.
  That $\apath$ is an access path also entails that $\avert$ is not visited again on $\apath$,
  hence that $\bvert\neq\avert$ and that $\avert$ occurs only once in $\abspre{\bvert}$,
  and that $\stringcon{\abspre{\avert}}{\avert}$, the abstraction prefix of the successor vertex of $\avert$ on $\apath$,
  is a prefix of the abstraction prefix of every vertex that is visited on $\apath$ after $\avert$. 
  
  Statements~(\ref{lem:laphotg:item:ii}) and (\ref{lem:laphotg:item:iii}) follow directly from statement~(\ref{lem:laphotg:item:i}).
  
  For showing (\ref{lem:laphotg:item:iv}), consider an access path 
  $\apath \funin \sroot = \bverti{0} \tgsucc \cdots \tgsucc \bverti{n}$ 
  that leads to a vertex $\bverti{n}\in\vertsof{\snlvar}$.
  If $i=0$, then there is no path that extends $\apath$ properly beyond $\bverti{n}$. 
  So suppose $i=1$, 
  and let $\bverti{n+1}\in\verts$ be such that $\bverti{n} \tgsucci{0} \bverti{n+1}$. 
  Then ($\snlvar$)$_1$ implies that  
  $\abspre{\bverti{n}} = \stringcon{\abspre{\bverti{n+1}}}{\bverti{n+1}}$,
  from which it follows by (\ref{lem:laphotg:item:i})  
  that 
  $\bverti{n+1}$ is visited already on $\apath$.
  Hence $\apath$ does not extend to a longer path that is again an access path.
\end{proof}

\begin{definition}[homomorphism, bisimulation]\label{def:homom:aplambdahotg}\normalfont
  Let $i\in\setexp{0,1}$.
  Let $\alaphotgi{1}$ and $\alaphotgi{2}$ 
  be \lambdaaphotg{s} over $\siglambdai{i}$ with
  $\alaphotgi{k} = \tuple{\vertsi{k},\svlabi{k},\svargsi{k},\srooti{k},\sabsprei{k}}$
  for $k\in\setexp{1,2}$.

  A \emph{homomorphism}, also called a \emph{functional bisimulation}, 
  from $\alaphotgi{1}$ to $\alaphotgi{2}$ 
  is a morphism from the structure
  $\tuple{\vertsi{1},\svlabi{1},\svargsi{1},\srooti{1},\sabsprei{1}}$
  to the structure
  $\tuple{\vertsi{2},\svlabi{2},\svargsi{2},\srooti{2},\sabsprei{2}}$,
  that is, a function 
  $ \sahom \funin \vertsi{1} \to \vertsi{2}$ 
  such that, for all $\avert\in\vertsi{1}$ 
  the conditions (labels), (arguments), and (roots) in 
  in \eqref{eq:def:homom} are satisfied, and additionally,
  for all $\avert\in\vertsi{1}$:
\begin{equation}\label{eq:def:homom:aplambdahotg}
\begin{aligned}
  \funap{\bar{\sahom}}{\absprei{1}{\avert}}
    & = \absprei{2}{\ahom{\avert}}
  & & & & & (\text{abstraction-prefix functions})
\end{aligned}
\end{equation}
where ${\bar{\sahom}}$ is the homomorphic extension of $\sahom$ to words over $\vertsi{1}$.
In this case we write $\alaphotgi{1} \funbisimi{\sahom} \alaphotgi{2}$,
or $\alaphotgi{2} \convfunbisimi{\sahom} \alaphotgi{1}$.
And we write
$\alaphotgi{1} \funbisim \alaphotgi{2}$,
or for that matter $\alaphotgi{2} \convfunbisim \alaphotgi{1}$,
if there is a homomorphism (a functional bisimulation) from $\alaphotgi{1}$ to $\alaphotgi{2}$.

A \emph{bisimulation} between  
$\alaphotgi{1}$ and $\alaphotgi{2}$ 
is a term graph 
$ \alaphotg = \tuple{\abisim,\svlab,\svargs,\sroot,\sScope}$ over $\asig$
with $\abisim \subseteq \vertsi{1}\times\vertsi{2}$ 
 and $\sroot = \pair{\srooti{1}}{\srooti{2}}$
such that
$ \alaphotgi{1} \convfunbisimi{\sproji{1}} \alaphotg \funbisimi{\sproji{2}} \alaphotgi{2}$
where $\sproji{1}$ and $\sproji{2}$ are projection functions, defined, for $i\in\setexp{1,2}$,
by $ \sproji{i} \funin \vertsi{1}\times\vertsi{2} \to \vertsi{i} $,
$\pair{\averti{1}}{\averti{2}} \mapsto \averti{i}$.
If there exists a homomorphism (a functional bisimulation) $\sahom$ from $\alaphotgi{1}$ to $\alaphotgi{2}$,
then we write $\alaphotgi{1} \funbisimi{\sahom} \alaphotgi{2}$
or $\alaphotgi{2} \convfunbisimi{\sahom} \alaphotgi{1}$,
or, dropping $\sahom$ as subscript, 
$\alaphotgi{1} \funbisim \alaphotgi{2}$ or $\alaphotgi{2} \convfunbisim \alaphotgi{1}$.
\end{definition}

The following proposition defines mappings between \lambdahotg{s} and
\lambdaaphotg{s} by which we establish a bijective correspondence between the
two classes. For both directions the underlying \lambdatg\ remains unchanged.
$\slhotgstolaphotgsi{i}$ derives an abstraction-prefix function $\sabspre$ from
a scope function by assigning to each vertex a word of its binders in the
correct nesting order.
$\slaphotgstolhotgsi{i}$ defines its scope function $\sScope$ by
assigning to each \lambda-vertex $\avert$ the set of vertices that have
$\avert$ in their prefix (along with $\avert$ since a vertex never has itself
in its abstraction prefix).

\begin{proposition}\label{prop:mappings:lhotgs:laphotgs}
  For each $i\in\setexp{0,1}$, the mappings $\slhotgstolaphotgsi{i}$ and $\slaphotgstolhotgsi{i}$ 
  are well-defined between the class of \lambdahotg{s} over $\siglambdai{i}$ 
  and the class of \lambdaaphotg{s} over $\siglambdai{i}$:
\begin{align}
& 
\left.
\begin{aligned}
  &
  \slhotgstolaphotgsi{i} \funin \classlhotgsi{i} \to \classlaphotgsi{i},
    \;\;
  \alhotg = \tuple{\verts,\svlab,\svargs,\sroot,\sScope} 
    \mapsto \lhotgstolaphotgsi{i}{\alhotg} \defdby \tuple{\verts,\svlab,\svargs,\sroot,\sabspre} 
  \\
  & \hspace*{22ex}
  \text{where }
  \sabspre \funin \verts\to\verts^*, \;\,
  \bvert \mapsto \averti{0}\mcdots\averti{n} \;
    \text{ if } \begin{aligned}[t]
                  & \binders{\bvert} \setminus \setexp{\bvert} = \setexp{\averti{0},\ldots,\averti{n}} \text{ and}
                  \\[-0.5ex]
                  & \Scope{\averti{n}} \subsetneqq \Scope{\averti{n-1}} \ldots \subsetneqq \Scope{\averti{0}}
                 \end{aligned}  
\end{aligned}
\right\} \label{eq1:prop:mappings:lhotgs:laphotgs}
\\[0.75ex]
&
\left.
\begin{aligned}
  &
  \slaphotgstolhotgsi{i} \funin \classlaphotgsi{i} \to \classlhotgsi{i},
    \;\;
  \alhotg = \tuple{\verts,\svlab,\svargs,\sroot,\sabspre} 
    \mapsto \lhotgstolaphotgsi{i}{\alhotg} \defdby \tuple{\verts,\svlab,\svargs,\sroot,\sScope} 
  \\
  & \hspace*{22ex}
  \text{where }
  \sScope \funin \vertsof{\sslabs}\to\powersetof{\verts}, \;\,
    \avert \mapsto \descsetexp{\bvert\in\verts}{\avert\text{ occurs in }\abspre{\bvert}} \cup \setexp{\avert}
\end{aligned}
\hspace*{1.7ex}\right\} \label{eq2:prop:mappings:lhotgs:laphotgs}
\end{align}
\end{proposition}

\begin{theorem}[correspondence of \lambdahotg{s} with \lambdaaphotg{s}]\label{thm:corr:lhotgs:laphotgs}
  For each $i\in\setexp{0,1}$ it holds that
  the mappings $\slhotgstolaphotgsi{i}$ in \eqref{eq1:prop:mappings:lhotgs:laphotgs}
  and $\slaphotgstolhotgsi{i}$ in \eqref{eq2:prop:mappings:lhotgs:laphotgs}
  are each other's inverse;
  thus they define a bijective correspondence between  
  the class of \lambdahotg{s} over $\siglambdai{i}$ 
  and the class of \lambdaaphotg{s} over~$\siglambdai{i}$.
  Furthermore, they preserve and reflect the sharing orders on $\classlhotgsi{i}$ and on $\classlaphotgsi{i}$:
  \begin{align*}
    (\forall \alhotgi{1},\alhotgi{2}\in\classlhotgsi{i}) \;\; \hspace*{-18ex} &&
      \alhotgi{1} \funbisim \alhotgi{2}
        \;\; &\Longleftrightarrow \;\;
      \lhotgstolaphotgsi{i}{\alhotgi{1}} \funbisim \lhotgstolaphotgsi{i}{\alhotgi{1}} 
    \\  
    (\forall \alaphotgi{1},\alaphotgi{2}\in\classlaphotgsi{i}) \;\; \hspace*{-18ex} &&  
       \laphotgstolhotgsi{i}{\alaphotgi{1}} \funbisim \laphotgstolhotgsi{i}{\alaphotgi{1}}
         \;\; &\Longleftrightarrow \;\;
       \alhotgi{1} \funbisim \alhotgi{2}
  \end{align*}
\end{theorem}

\begin{example}\label{ex:corr:lhotgs:laphotgs}\normalfont
The \lambdahotg{s} in Fig.~\ref{fig:lambdahotg} correspond to the \lambdaaphotg{s} in Fig.~\ref{fig:lambdaaphotg}
via the mappings $\slhotgstolaphotgsi{i}$ and $\slaphotgstolhotgsi{i}$ as follows: 
$
~~~~~~ \lhotgstolaphotgsi{i}{\alhotgi{0}} = \alaphotgi{0}' \, ,
~~~~~~ \lhotgstolaphotgsi{i}{\alhotgi{1}} = \alaphotgi{1}' \, ,
~~~~~~ \laphotgstolhotgsi{i}{\alaphotgi{0}} = \alhotgi{0}' \, ,
~~~~~~ \lhotgstolaphotgsi{i}{\alaphotgi{1}} = \alhotgi{1}' \, 
$.
\end{example}


For \lambdahotg{s}\ over the signature $\siglambdai{0}$ (that is,
without variable back-links) essential binding information is lost when looking
only at the underlying term graph, to the extent that \lambdaterms\ cannot be
unambiguously represented anymore. For instance the \lambdahotg{s} that
represent the \lambdaterms\ $\labs{\avar\bvar}{\lapp{\avar}{\bvar}}$ and
$\labs{\avar\bvar}{\lapp{\avar}{\avar}}$ have the same
underlying term graph. The same holds for \lambdaaphotg{s}.


This is not the case for \lambdahotg{s} (\lambdaaphotg{s}) over
$\siglambdai{1}$, because the abstraction vertex to which a variable-occurrence
vertex belongs is uniquely identified by the back-link. 
This is the reason why the following notion is only defined for the signature
$\siglambdai{1}$.

\begin{definition}[\lambdatg\ over $\siglambdai{1}$]\label{def:classltgs}\normalfont
  A term graph $\atg$ over $\siglambdai{1}$ is called a \emph{\lambdatg} over $\siglambdai{1}$
  if $\atg$ admits a correct ab\-strac\-tion-pre\-fix function. 
  By $\classltgsi{1}$ we denote the class of \lambdatg{s} over $\siglambdai{1}$.
\end{definition}
 
In the rest of this section we examine, and then dismiss, a naive approach to
implementing functional bisimulation on \lambdahotg{s} or \lambdaaphotg{s},
which is to apply the homomorphism on the underlying term graph, hoping that
this application would simply extend to the \lambdahotg\ (\lambdaaphotg)
without further ado. We demonstrate that this approach fails, concluding that a
faithful first-order implementation of functional bisimulation must not be
negligent of the scoping information.

\begin{definition}[scope- and \absprefix-forgetful mappings]\label{def:forgetful}\normalfont
  Let $i\in\setexp{0,1}$. 
  The \emph{scope-forgetful mapping}~$\sscopeforgetfullhotgsi{i}$ 
  and the \emph{\absprefix-for\-get\-ful mapping}~$\sabspreforgetfullaphotgsi{1}$
  map \lambdahotg{s} in $\classltgsi{i}$, and respectively, \lambdahotg{s} in $\classltgsi{i}$
  to their underlying term graphs: 
  \begin{center}
    $\sscopeforgetfullhotgsi{i} \funin \classlhotgsi{i} \to \classtgssiglambdai{i}\, , \;\; 
       \tuple{\verts,\svlab,\svargs,\sroot,\sScope} \mapsto \tuple{\verts,\svlab,\svargs,\sroot}$
    \\[0.5ex]
    $\sabspreforgetfullaphotgsi{i} \funin \classlaphotgsi{i} \to \classtgssiglambdai{i}\, , \;\; 
       \tuple{\verts,\svlab,\svargs,\sroot,\sabspre} \mapsto \tuple{\verts,\svlab,\svargs,\sroot}$
  \end{center}  
\end{definition}

\begin{definition}\normalfont
  Let $\alhotg$ be a \lambdahotg\ over $\siglambdai{i}$ for $i\in\setexp{0,1}$ 
  with underlying term graph $\scopeforgetfullhotgsi{i}{\alhotg}$.
  And suppose that $\scopeforgetfullhotgsi{i}{\alhotg} \funbisimi{\sahom} \atgacc$ holds
  for a term graph $\atgacc$ over $\siglambdai{i}$
  and a functional bisimulation $\sahom$.
  We say that $\sahom$ \emph{extends to a functional bisimulation on} $\alhotg$
  if $\atgacc$ can be endowed with a scope function to obtain a \lambdahotg\ $\alhotgacc$ 
  with $\scopeforgetfullhotgsi{i}{\alhotgacc} = \atgacc$ 
  and such that it holds $\alhotg \funbisimi{\sahom} \alhotgacc$. 
  
  We say that a class $\aclass$ of \lambdahotg{s} is 
  \emph{closed under functional bisimulations on the underlying term graphs}
  if for every $\alhotg\in\aclass$
    and for every homomorphism $\sahom$ on the term graph underlying $\alhotg$ 
    that witnesses $\scopeforgetfullhotgsi{i}{\alhotg} \funbisimi{\sahom} \atgacc$ for a term graph $\atgacc$
    there exists $\alhotgacc\in\aclass$ with $\scopeforgetfullhotgsi{i}{\alhotgacc} = \atgacc$
    such that $\alhotg \funbisimi{\sahom} \alhotgacc$ holds, that is,
    $\sahom$ is also a homorphism between $\alhotg$ and $\alhotgacc$.
    
  
  These notions are also extended, by analogous stipulations,
  to \lambdaaphotg{s} over $\siglambdai{i}$ for $i\in\setexp{0,1}$ and their underlying term graphs. 
\end{definition}

%
%

\begin{proposition}\label{prop:no:extension:funbisim:siglambda1}
  Neither the class $\classlhotgsi{1}$ of \lambdahotg{s} nor the class $\classlaphotgsi{1}$ of \lambdaaphotg{s} 
  is closed under functional bisimulations on the underlying term graphs.

\end{proposition}

\begin{proof}
  In view of Thm.~\ref{thm:corr:lhotgs:laphotgs} it suffices to show the statement for $\classlhotgsi{1}$.
  We show that not every functional bisimulation on the term graph underlying 
  a \lambdahotg\ over $\siglambdai{1}$ 
  extends to a functional bisimulation on the higher-order term graphs.
  Consider the following term graphs $\atgi{0}$ and $\atgi{1}$ over $\siglambdai{1}\hspace*{1pt}$
  (at first, please ignore the scope shading):
  \graphs
  {
          \vcentered{$\atgi{1}$:~~} \vcentered{\includegraphics[scale=0.80]{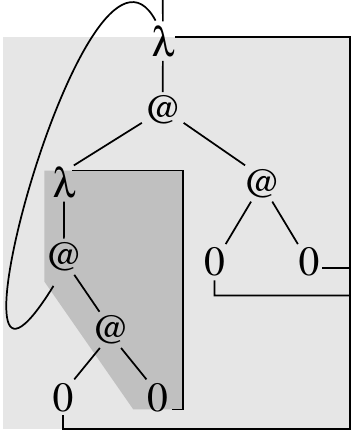}}\hfill
          \graph{\atgi{0}}{no_extension_funbisim_siglambda1_g0}
  }%
  There is an obvious homomorphism $\sahom$ that witnesses 
  \mbox{$\atgi{1} \funbisimi{\sahom} \atgi{0}$}.
  Both of these term graphs extend to \lambdahotg{s} by suitable scope functions
  (one possibility per term graph is indicated by the scope shadings above; $\atgi{1}$ actually admits two possibilities). 
  However, $\sahom$ does not extend to any of the \lambdahotg{s} $\alhotgi{1}$ and $\alhotgi{0}$ 
  that extend $\atgi{1}$ and $\atgi{0}$, respectively. 
\end{proof}

The next proposition is merely a reformulation of Prop.~\ref{prop:no:extension:funbisim:siglambda1}.

\begin{proposition}\label{prop:forgetful}
  The scope-forgetful mapping $\sscopeforgetfullhotgsi{1}$ on $\classlhotgsi{1}$ 
  and the \absprefix-for\-get\-ful mapping $\sabspreforgetfullaphotgsi{1}$ on $\classlaphotgsi{1}$ 
  preserve, but do not reflect, the sharing orders on these classes. In particular:
  \begin{align*}
    (\forall \alaphotgi{1},\alaphotgi{2}\in\classlaphotgsi{1}) \;\; \hspace*{-14ex} &&
       \alaphotgi{1} \funbisim \alaphotgi{2}
         \;\: & \Longrightarrow\;\:
       \abspreforgetfullaphotgsi{1}{\alaphotgi{1}} \funbisim \abspreforgetfullaphotgsi{1}{\alaphotgi{2}}
    \\[-0.25ex]
    (\exists \alaphotgi{1},\alaphotgi{2}\in\classlaphotgsi{1}) \;\; \hspace*{-14ex} &&
       \alaphotgi{1} \not\funbisim \alaphotgi{2}
         \;\: & \logand \;\:
       \abspreforgetfullaphotgsi{1}{\alaphotgi{1}} \funbisim \abspreforgetfullaphotgsi{1}{\alaphotgi{2}}
  \end{align*}
  
\end{proposition}

As a consequence of this proposition it is not possible to faithfully implement functional
bisimulation on \lambdahotg{s} and \lambdaaphotg{s} by only considering the underlying term graphs,
and in doing so neglecting%
         \footnote{
           In the case of $\siglambdai{1}$ implicit information about possible scopes is being kept,
           due to the presence of back-links from variable occurrence vertices to abstraction vertices.
           But this is not enough for reflecting the sharing order under the forgetful mappings.}
the scoping information          
from the scope function, or respectively, from the abstraction prefix function.
In order to yet be able to implement functional bisimulation of
\lambdahotg{s} and \lambdaaphotg{s} in a first-order setting, in the next
section we introduce a class of first-order term graphs that accounts for
scoping by means of scope delimiter vertices.


\section{$\lambda$-Term-Graphs with Scope Delimiters}
  \label{sec:ltgs}

For all $i\in\setexp{0,1}$ and $j\in\setexp{1,2}$ we define the extensions
$\siglambdaij{i}{j} \defdby \siglambda \cup \setexp{\snlvar, \snlvarsucc}$ 
of the signature $\siglambda$
where $\arity{\snlvar} = i$ and $\arity{\snlvarsucc} = j\,$,
and we denote the class of term graphs over signature $\siglambdaij{i}{j}$ by $\classtgssiglambdaij{i}{j}$.

Let $\atg$ be a term graph with vertex set $\verts$ over a signature extending $\siglambdaij{i}{j}$ 
for $i\in\setexp{0,1}$ and $j\in\setexp{1,2}$.
We denote by $\vertsof{\snlvarsucc}$ the subset of $\verts$ consisting of all vertices
with label $\snlvarsucc$, which are called the \emph{delimiter vertices} of $\atg$.
Delimiter vertices signify the end of an `extended scope' \cite{grab:roch:2012}. 
They are analogous to occurrences of function symbols $\snlvarsucc$
in representations of \lambdaterms\ in a nameless 
de-Bruijn index \cite{de-bruijn:72} form in which Dedekind numerals
based on $\snlvar$ and the successor function symbol $\snlvarsucc$ are used 
(this form is due to Hendriks and van~Oostrom, see also \cite{oost:looi:zwit:2004},
 and is related to their end-of-scope symbol $\adbmal$ \cite{hend:oost:2003}).

Analogously as for the classes $\classlhotgsi{i}$ and $\classlaphotgsi{i}$,
the index $i$ will determine whether in correctly formed \lambdatg{s} (defined below)
variable vertices have back-links to the corresponding abstraction. 
Here additionally scope-delimiter vertices have such back-links (if $j=2$) or not (if $j=1$). 

\begin{definition}[correct abstraction-prefix function for $\siglambdaij{i}{j}$-term-graphs]%
    \label{def:abspre:function:siglambdaij}\normalfont
  Let $\atg = \tuple{\verts,\svlab,\svargs,\sroot}$ be a $\siglambdaij{i}{j}$\nb-term-graph
  for an $i\in\setexp{0,1}$ and an $j\in\setexp{1,2}$.
  
  A function $\sabspre \funin \verts \to \verts^*$
  from vertices of $\atg$ to words of vertices 
  is called an \emph{abstraction-prefix function} for $\atg$.
  Such a function is called \emph{correct} if
  for all $\bvert,\bverti{0},\bverti{1}\in\verts$ and $k\in\setexp{0,1}$ it holds:
  \begin{align*}
      \;\; & \Rightarrow \;\;
    \abspre{\sroot} = \emptyword  
    &
    (\text{root})
    \displaybreak[0]\\
    \bvert\in\vertsof{\sslabs} 
      \;\logand\;
    \bvert \tgsucci{0} \bvertbp{0}{}
      \;\; & \Rightarrow \;\;
    \abspre{\bvertbp{0}{}} = \abspre{\bvert} \bvert 
    & 
    (\sslabs)
  \displaybreak[0]\\
    \bvert\in\vertsof{\sslapp}
      \;\logand\;
    \bvert \tgsucci{k} \bvertbp{k}{}
      \;\; & \Rightarrow \;\;
    \abspre{\bvertbp{k}{}} = \abspre{\bvert} 
    & 
    (\sslapp)
  \displaybreak[0]\\  
    \bvert\in\vertsof{\snlvar}
      \phantom{\;\;\logand\;\bvert \tgsucci{0} \bvertacc}
        \;\; & \Rightarrow \;\;
     \abspre{\bvert} \neq \emptyword   
    & 
    (\snlvar)_0
  \displaybreak[0]\\
    \bvert\in\vertsof{\snlvar}
      \;\logand\;
    \bvert \tgsucci{0} \bvertbp{0}{}
      \;\; & \Rightarrow \;\;
    \begin{aligned}[c]  
      & 
        \bvertbp{0}{}\in\vertsof{\sslabs}
      \;\logand\;
      \abspre{\bvertbp{0}{}} \bvertbp{0}{} = \abspre{\bvert} 
    \end{aligned}
    & 
    (\snlvar)_1 
    \displaybreak[0]\\
    \bvert\in\vertsof{\snlvarsucc}
      \;\logand\;
    \bvert \tgsucci{0} \bvertbp{0}{}
        \;\; & \Rightarrow \;\;
    \abspre{\bvertbp{0}{}} \avert
      =
    \abspre{\bvert}
      \;\;   
    \text{for some $\avert\in\verts$} \hspace*{1ex}
    &
    (\snlvarsucc)_1 
  \displaybreak[0]\\
    \bvert\in\vertsof{\snlvarsucc}
      \;\logand\;
    \bvert \tgsucci{1} \bvertbp{1}{}
        \;\; & \Rightarrow \;\;
    \bvertbp{1}{} \in \vertsof{\sslabs}
      \;\logand\;        
    \abspre{\bvertbp{1}{}} \bvertbp{1}{}
      =
    \abspre{\bvert}
    &
    (\snlvarsucc)_2  
  \end{align*}%
  Note that analogously as in Def.~\ref{def:lambdahotg} and in Def.~\ref{def:aplambdahotg},
  if $i=0$, then $(\snlvar)_1$ is trivially true and hence superfluous,
  and if $i=1$, then $(\snlvar)_0$ is redundant, because it follows from $(\snlvar)_1$ in this case. 
  Additionally, if $j=1$, then $(\snlvarsucc)_2$ is trivially true and therefore superfluous.         
\end{definition}

\begin{definition}[{\lambdatg\ over $\siglambdaij{i}{j}$}]%
    \label{def:lambdatg:siglambdaij}\normalfont
  Let $i\in\setexp{0,1}$ and $j\in\setexp{1,2}$.
  A \emph{\lambdatg\ (with scope-de\-li\-mi\-ters)} over $\siglambdaij{i}{j}$
  is a $\siglambdaij{i}{j}$\nb-term-graph
  that admits a correct ab\-strac\-tion-pre\-fix function. 
  The class of \lambdatg{s} over $\siglambdaij{i}{j}$ is denoted by $\classltgsij{i}{j}$.
\end{definition}

See Fig.~\ref{fig:lambdatg} for examples,
that, as we will see, correspond to the ho-term-graphs in Fig.~\ref{fig:lambdahotg} and in Fig.~\ref{fig:lambdaaphotg}.

\begin{figure}[t]
\graphs{
  \graph{\atgi{0}}{running_snodes_eager}
  \graph{\atgi{1}}{running_snodes_non_eager}
}
\vspace*{-2ex}
\caption{\label{fig:lambdatg}%
The \protect\lambdatg{s} corresponding to the \protect\lambdaaphotg{s} from
Fig~\ref{fig:lambdaaphotg} and the \protect\lambdahotg{s} from Fig~\ref{fig:lambdahotg}.
A precise formulation of this correspondence is given in
Example~\ref{ex:corr:laphotgs:ltgs}.
}
\end{figure}

\begin{lemma}\label{lem:ltgs}
  Let $i\in\setexp{0,1}$ and $j\in\setexp{1,2}$,
  and let $\altg = \tuple{\verts,\svlab,\svargs,\sroot}$ be a \lambdatg\ over $\siglambdaij{i}{j}$. 
  Then the statements \eqref{lem:laphotg:item:i}--\eqref{lem:laphotg:item:iii} in Lemma~\ref{lem:laphotgs} hold,
  and additionally: 
  \begin{enumerate}[(i)]\setcounter{enumi}{3}\itemizeprefs
    \item\label{lem:ltg:item:iv}
      Access paths may end in vertices in $\vertsof{\snlvar}$,
      but only pass through vertices in $\vertsof{\sslabs} \cup \vertsof{\sslapp} \cup \vertsof{\snlvarsucc}$,
      and depart from vertices in $\vertsof{\snlvarsucc}$ only via indexed edges $\stgsuccis{0}{\snlvarsucc}$.
    \item\label{lem:ltg:item:v}
      There exists precisely one correct ab\-strac\-tion-pre\-fix function on $\altg$.
  \end{enumerate}
\end{lemma}

\begin{proof}
  That also here statements (\ref{lem:laphotg:item:i})--(\ref{lem:laphotg:item:iii}) in Lemma~\ref{lem:laphotgs} hold,
  and that statement~(\ref{lem:ltg:item:iv}) holds,
  can be shown analogously as in the proof of the respective items of Lemma~\ref{lem:laphotgs}. 
  For (\ref{lem:ltg:item:v}) it suffices to observe
  that if $\sabspre$ is a correct \absprefix\ function for $\altg$,
  then, for all $\bvert\in\verts$, the value $\abspre{\bvert}$ of $\sabspre$ at $\bvert$
  can be computed by choosing an arbitrary access path $\apath$ from $\sroot$ to $\bvert$
  and using the conditions $(\sslabs)$, $(\sslapp)$, and $(\snlvarsucc)_0$  
  to determine in a stepwise manner the values of $\sabspre$ at the vertices that are visited on $\apath$.
  Hereby note that in every transition along an edge on $\apath$ the length of the abstraction prefix only changes by at most 1. 
\end{proof}

Now we define a precise relationship between \lambdatg{s} and \lambdaaphotg{s}
via translation mappings between these classes:
\begin{description}\itemizeprefs
\item[The mapping $\slaphotgstoltgsij{i}{j}$ {\normalfont (see Prop.~\ref{prop:mappings:laphotgs:to:ltgs})}] 
  produces a \lambdatg\ for any given \lambdaaphotg\ by adding to the
  original set of vertices a number of delimiter vertices at the appropriate
  places. That is, at every position where the abstraction prefix decreases by $n$
  elements, $n$ $\snlvarsucc$-vertices are inserted. In the image, the original
  abstraction prefix is retained as part of the vertices. This can be considered
  intermittent information used for the purpose of defining the edges of the image.
\item[The mapping $\sltgstolaphotgsij{i}{j}$ {\normalfont (see Prop.~\ref{prop:mappings:ltgs:to:laphotgs})}] 
  back to \lambdaaphotg{s} is
  simpler because it only has to erase the $\snlvarsucc$-vertices, and add the correct abstraction prefix
  that exists for the \lambdatg\ to be translated.
\end{description}
 

\begin{proposition}\label{prop:mappings:laphotgs:to:ltgs}
  Let $i\in\setexp{0,1}$ and $j\in\setexp{1,2}$.
  The mapping $\slaphotgstoltgsij{i}{j}$ defined below 
  is well-defined between the class of \lambdatg{s} over $\siglambdaij{i}{j}$
  and the class of \lambdaaphotg{s} over $\siglambdai{i}$:
\begin{align*}
&    
\begin{aligned}
  &
  \slaphotgstoltgsij{i}{j} \funin \classlaphotgsi{i} \to \classltgsij{i}{j},
    \;\;
  \alhotg = \tuple{\verts,\svlab,\svargs,\sroot,\sabspre} 
    \mapsto \laphotgstoltgsij{i}{j}{\alhotg} \defdby \tuple{\vertsacc,\svlabacc,\svargsacc,\srootacc} 
\end{aligned}
\end{align*}
where:
\begin{equation*}
\begin{gathered}  
  \vertsacc \defdby
    \descsetexp{\pair{\bvert}{\abspre{\bvert}}}{\bvert\in\verts}
      \cup
    \descsetexpBig{\tuple{\bvert,k,\bvertacc,\apre}}{\bvert,\bvertacc\in\verts,\,
                                              \bvert \tgsucci{k} \bvertacc,\,
                                              \begin{aligned}[c]
                                                & \vertsof{\bvert} = \sslabs
                                                    \:\logand\:
                                                  \abspre{\bvertacc} \prelt \apre \prele \stringcon{\abspre{\bvert}}{\bvert}
                                                \\[-0.75ex]
                                                & \logor\:
                                                  \vertsof{\bvert} = \sslapp
                                                    \:\logand\:
                                                  \abspre{\bvertacc} \prelt \apre \prele \abspre{\bvert}
                                              \end{aligned}} 
  \\
  \srootacc \defdby \pair{\sroot}{\emptyword} 
  \hspace*{10ex} 
  \begin{gathered}[t]
  \svlabacc \funin \vertsacc \to \siglambdaij{i}{j} \, , \;
    \begin{aligned}[t]
      \pair{\bvert}{\abspre{\bvert}} & \mapsto \vlabacc{\pair{\bvert}{\abspre{\bvert}}} \defdby \vlab{\bvert}
      \\[-0.5ex]
      \tuple{\bvert,k,\bvertacc,\apre} & \mapsto \vlabacc{\bvert,k,\bvertacc,\apre} \defdby \snlvarsucc
    \end{aligned}
  \end{gathered}  
\end{gathered}    
\end{equation*}
and $\svargsacc \funin \vertsacc \to (\vertsacc)^*$ is defined such that 
for the induced indexed successor relation $\stgsuccacci{(\cdot)}$ it holds: 
\begin{align*}
  &
  \bvert \tgsucci{k} \bverti{k}
    \:\logand\:
  \nodels{\bvert}{k} = 0
    \;\;\Longrightarrow\;\;
  \pair{\bvert}{\abspre{\bvert}} \tgsuccacci{k} \pair{\bverti{k}}{\abspre{\bverti{k}}} 
  \displaybreak[0]\\
  &
  \bvert \tgsucci{0} \bverti{0}
    \:\logand\:  
  \nodels{\bvert}{0} > 0
    \:\logand\:
  \vlab{\bvert} = \sslabs
    \:\logand\:
  \abspre{\bvert} = \stringcon{\stringcon{\abspre{\bverti{0}}}{\avert}}{\apre} 
    \\[-0.5ex]
    & \hspace*{6ex} 
    \;\;\Longrightarrow\;\;
      \pair{\bvert}{\abspre{\bvert}} \stgsuccacci{0} \tuple{\bvert,0,\bverti{0},\stringcon{\abspre{\bvert}}{\bvert}}
        \:\logand\:
      \tuple{\bvert,0,\bverti{0},\abspre{\bverti{0}}\avert} \stgsuccacci{0} \pair{\bverti{0}}{\abspre{\bverti{0}}}  
  \displaybreak[0]\\
  &   
  \bvert \tgsucci{k} \bverti{k} 
    \:\logand\:  
  \nodels{\bvert}{k} > 0
    \:\logand\:
  \vlab{\bvert} = \sslapp
    \:\logand\:
  \abspre{\bvert} = \stringcon{\stringcon{\abspre{\bverti{k}}}{\avert}}{\apre} 
    \\[-0.5ex]
    & \hspace*{6ex} 
    \;\;\Longrightarrow\;\;
      \pair{\bvert}{\abspre{\bvert}} \stgsuccacci{k} \tuple{\bvert,k,\bverti{k},\abspre{\bvert}}
        \:\logand\:
      \tuple{\bvert,k,\bverti{k},\abspre{\bverti{k}}\avert} \stgsuccacci{0} \pair{\bverti{k}}{\abspre{\bverti{k}}}
  \displaybreak[0]\\ 
  &  
  \bvert \tgsucci{k} \bverti{k}
    \:\logand\:  
  \nodels{\bvert}{k} > 0
    \:\logand\:
  \tuple{\bvert,k,\bverti{k},\stringcon{\apre}{\avert}}, \,
  \tuple{\bvert,k,\bverti{k},\apre} \in\vertsacc
    \;\;\Longrightarrow\;\;
  \tuple{\bvert,k,\bverti{k},\stringcon{\apre}{\avert}} \tgsuccacci{0} \tuple{\bvert,k,\bverti{k},\apre}
  \\
  &
  \bvert \tgsucci{k} \bverti{k}
    \:\logand\:  
  \nodels{\bvert}{k} > 0
    \:\logand\:
  \tuple{\bvert,k,\bverti{k},\stringcon{\apre}{\avert}} \in\vertsacc
    \:\logand\:
  j = 2  
    \;\;\Longrightarrow\;\;
  \tuple{\bvert,k,\bverti{k},\stringcon{\apre}{\avert}} \tgsuccacci{1} \pair{\bverti{k}}{\abspre{\bverti{k}}}
\end{align*}
for all $\bvert,\bverti{0},\bverti{1},\avert\in\verts$, $k\in\setexp{0,1}$, $\apre\in\verts^*$,
and where the function $\snodels$ is defined as:
\begin{equation*}
  \nodels{\bvert}{k} 
    \defdby
  \begin{cases}
    \length{\abspre{\bvert}} - \length{\abspre{\bvertacc}} 
      & \text{if } \bvert\in\vertsof{\sslapp} \logand \bvert \stgsucci{k} \bvertacc 
    \\[-0.5ex]
    \length{\abspre{\bvert}} + 1 - \length{\abspre{\bvertacc}} 
      & \text{if } \bvert\in\vertsof{\sslabs} \logand \bvert \stgsucci{k} \bvertacc
    \\[-0.5ex]
    0 & \text{otherwise}   
  \end{cases}  
\end{equation*}
\end{proposition}

\begin{proposition}\label{prop:mappings:ltgs:to:laphotgs}
  Let $i\in\setexp{0,1}$ and $j\in\setexp{1,2}$.
  The mapping $\sltgstolaphotgsij{i}{j}$ defined below 
  is well-defined between the class of \lambdatg{s} over $\siglambdaij{i}{j}$ 
  and the class of \lambdaaphotg{s} over $\siglambdai{i}$:
\begin{align*}
& 
\begin{aligned}
  &
  \sltgstolaphotgsij{i}{j} \funin \classltgsij{i}{j} \to \classlaphotgsi{i},
    \;\;
  \altg = \tuple{\verts,\svlab,\svargs,\sroot} 
    \mapsto \ltgstolaphotgsij{i}{j}{\altg} \defdby \tuple{\vertsacc,\svlabacc,\svargsacc,\sroot',\sabspre'} 
  \\
  & \hspace*{15ex}
  \text{where }
  \begin{aligned}[t]
    & \vertsacc \defdby \vertsof{\sslabs} \cup \vertsof{{\sslapp}} \cup \vertsof{\snlvar}, \:
      \svlabacc \defdby \srestrictfunto{\svlab}{\vertsacc}, \:
      \sroot' \defdby \sroot, \:
      \\[-0.5ex]
    & \svargsacc \funin \vertsacc \to (\vertsacc)^*
       \text{ so that for the induced indexed succ.\ relation $\stgsuccacci{(\cdot)}$:} 
      \\[-0.5ex]
    & \hspace*{6ex}
        \averti{0} \tgsuccacci{k} \averti{1} 
          \;\funin\,\Leftrightarrow\;
        \averti{0} \binrelcomp{\stgsucci{k}}{\stgsuccisstar{0}{\snlvarsucc}} \averti{1} 
        \hspace*{3ex}  
        \text{(for all $\averti{0},\averti{1}\in\vertsacc$, $k\in\setexp{0,1}$)} \; 
      \\[-0.5ex]
    & \sabspre' \defdby \srestrictfunto{\sabspre}{\vertsacc}
      \text{ for the correct \absprefix\ function $\sabspre$ for $\altg$.} \,   
  \end{aligned}
\end{aligned}
\end{align*}
\end{proposition}

\begin{theorem}[correspondence between \lambdaaphotg{s} with \lambdatg{s}]\label{thm:corr:laphotgs:ltgs}
  Let $i\in\setexp{0,1}$ and $j\in\setexp{1,2}$.
  The mappings $\sltgstolaphotgsij{i}{j}$ from Prop.~\ref{prop:mappings:ltgs:to:laphotgs} 
  and $\slaphotgstoltgsij{i}{j}$ from Prop.~\ref{prop:mappings:laphotgs:to:ltgs}
  define a correspondence 
  between the classes of \lambdatg{s} over $\siglambdaij{i}{j}$ and of \lambdaaphotg{s} over $\siglambdai{i}$ 
  with the following properties:
\begin{enumerate}[(i)]\itemizeprefs
  \item 
    $\scompfuns{\sltgstolaphotgsij{i}{j}}{\slaphotgstoltgsij{i}{j}} = \sidfunon{\classlaphotgsi{i}}$.
  \item 
    For all $\altg\in\classltgsij{i}{j}\,$: \mbox{ } 
    $(\compfuns{\slaphotgstoltgsij{i}{j}}{\sltgstolaphotgsij{i}{j})}{\altg}
           \Sfunbisim
         \altg$.
  \item $\sltgstolaphotgsij{i}{j}$  and $\slaphotgstoltgsij{i}{j}$ 
    preserve and reflect the sharing orders on $\classlaphotgsi{i}$ and on $\classltgsij{i}{j}$:
    \vspace*{-0.25ex}
  \begin{align*}
    (\forall \alaphotgi{1},\alaphotgi{2}\in\classlaphotgsi{i}) \;\; \hspace*{-14ex} &&
       \alaphotgi{1} \funbisim \alaphotgi{2}
       \;\: & \Longleftrightarrow\;\:
     \laphotgstoltgsij{i}{j}{\alaphotgi{1}} \funbisim \laphotgstoltgsij{i}{j}{\alaphotgi{2}}
    \\  
    (\forall \altgi{1},\altgi{2}\in\classltgsij{i}{j}) \;\; \hspace*{-14ex} && 
       \ltgstolaphotgsij{i}{j}{\altgi{1}} \funbisim \ltgstolaphotgsij{i}{j}{\altgi{2}}
          \;\: & \Longleftrightarrow\;\:
       \altgi{1} \funbisim \altgi{2}
  \end{align*}
\end{enumerate}
\end{theorem}

\begin{example}\label{ex:corr:laphotgs:ltgs}\normalfont
The \lambdaaphotg{s} in Fig.~\ref{fig:lambdaaphotg} correspond to the \lambdaaphotg{s} in Fig.~\ref{fig:lambdatg}
via the mappings $\slaphotgstoltgsij{i}{j}$ and $\sltgstolaphotgsij{i}{j}$ as follows: 
$
~~~~~~ \laphotgstoltgsij{i}{j}{\alaphotgi{0}} = \altgi{0}' \, ,
~~~~~~ \laphotgstoltgsij{i}{j}{\alaphotgi{1}} = \altgi{1}' \, ,
~~~~~~ \ltgstolaphotgsij{i}{j}{\altgi{0}} = \alaphotgi{0}' \, ,
~~~~~~ \ltgstolaphotgsij{i}{j}{\altgi{1}} = \alaphotgi{1}' \, 
$.
\end{example}

\begin{remark}\label{rem:corr:laphotgs:ltgs:nobij}\normalfont
The correspondence in Theorem~\ref{thm:corr:laphotgs:ltgs} is not a bijection since $\sltgstolaphotgsij{i}{j}$ is not injective.
This can be seen for the following graphs (here with $i=0$ and $j=1$) where we have $\sltgstolaphotgsij01(\atg) = \alhotg = \sltgstolaphotgsij01(\atg')$:
\graphs{
  \graph{\atg}{rem_corr_laphotgs_ltgs_nobij_snodes_shared}
  \graph{\alhotg}{rem_corr_laphotgs_ltgs_nobij_aphotg}
  \graph{\atg'}{rem_corr_laphotgs_ltgs_nobij_snodes_unshared}
}
Obviously \lambdaaphotg{s} are not capable of reproducing the different degrees of $S$-sharing.
\end{remark}

\section{Not closed under bisimulation and functional bisimulation}
  \label{sec:not:closed}

In this section we collect all negative results concerning closedness under bisimulation
and functional bisimulation for the classes of \lambdatg{s} as introduced in the previous section.

\begin{proposition}\label{prop:lambdatgs:not:closed:under:bisim}
  None of the classes 
  $\classltgsi{1}$ and\/ $\classltgsij{i}{j}$, for\/ $i\in\setexp{0,1}$ and $j\in\setexp{1,2}$, of \lambdatg{s}
  are closed under bisimulation.
\end{proposition}

This proposition is an immediate consequence of the next one, which can be viewed 
as a refinement, because it formulates non-closedness of classes of \lambdatg{s}
under specializations of bisimulation,
namely for functional bisimulation (under which some classes are not closed), 
and for converse functional bisimulation (under which none of the classes considered here is closed).

\begin{proposition}\label{prop:lambdatgs:not:closed:under:funbisim:convfunbisim}
  The following statements hold:
  \begin{enumerate}[(i)]\itemizeprefs
    \item{}\label{prop:lambdatgs:not:closed:under:funbisim:convfunbisim:item:i}
      None of the classes 
      $\classltgsij{\pmb{0}}{j}$ for\/ $j\in\setexp{1,2}$
      of \lambdatg{s} are closed under functional bisimulation $\sfunbisim$, or under
      converse functional bisimulation $\sconvfunbisim\,$.   
    \item{}\label{prop:lambdatgs:not:closed:under:funbisim:convfunbisim:item:ii}
      None of the classes 
      $\classltgsi{\pmb{1}}$ and\/ $\classltgsij{\pmb{1}}{j}$ for\/ $j\in\setexp{1,2}$
      of \lambdatg{s} are closed under converse functional bisimulation.  
%
%
    \item{}\label{prop:lambdatgs:not:closed:under:funbisim:convfunbisim:item:iii} 
      The class $\classltgsij{\pmb{1}}{\pmb{1}}$ of \lambdatg{s} 
      is not closed under functional bisimulation. 
    \item{}\label{prop:lambdatgs:not:closed:under:funbisim:convfunbisim:item:iv}  
      The class $\classltgsij{\pmb{1}}{\pmb{2}}$ of \lambdatg{s} 
      is not closed under functional bisimulation. 
  \end{enumerate}
\end{proposition}

\begin{proof}
  For showing (\ref{prop:lambdatgs:not:closed:under:funbisim:convfunbisim:item:i}),
  let $\bsig$ be one of the signatures 
                                       $\siglambdaij{0}{j}$. 
  Consider the following term graphs over $\bsig$:
  \graphs
  {
      \graph{\atgi{2}}{lambdatgs_not_closed_under_funbisim_convfunbisim_item_i_g2}
      \graph{\atgi{1}}{lambdatgs_not_closed_under_funbisim_convfunbisim_item_i_g1}
      \graph{\atgi{0}}{lambdatgs_not_closed_under_funbisim_convfunbisim_item_i_g0}
  }%
  Note that $\atgi{2}$ represents the syntax tree of the nameless de-Bruijn-index notation
  $\lapp{(\sslabs\snlvar)}{(\sslabs\snlvar)}$
  for the \lambdaterm\ $\lapp{(\labs{\avar}{\avar})}{(\labs{\avar}{\avar})}$. 
  Then it holds: $\atgi{2} \funbisim \atgi{1} \funbisim \atgi{0}$.
  But while $\atgi{2}$ and $\atgi{0}$ admit correct ab\-strac\-tion-pre\-fix functions over $\bsig$ (nestedness of the implicitly defined scopes, here shaded),
  and consequently are \lambdatg{s} over $\bsig$, this is not the case for $\atgi{1}$  (overlapping scopes).
  Hence the class of \lambdatg{s} over $\bsig$ is closed neither under functional bisimulation nor under converse functional bisimulation.
  
  For showing (\ref{prop:lambdatgs:not:closed:under:funbisim:convfunbisim:item:ii}),
  let $\bsig$ be one of the signatures $\siglambdai{1}$ and $\siglambdaij{1}{j}$.
  Consider the term graphs over $\bsig$:
  \graphs
  {
          \graph{\atgi{1}'}{lambdatgs_not_closed_under_funbisim_convfunbisim_item_ii_g1}
          \graph{\atgi{0}'}{lambdatgs_not_closed_under_funbisim_convfunbisim_item_ii_g0}
  }
  Then it holds: $\atgi{1}' \funbisim \atgi{0}'$.
  But while $\atgi{0}'$ admits a correct ab\-strac\-tion-pre\-fix function,
  and therefore is a \lambdatg, over $\bsig$, this is not the case for $\atgi{1}'$ (due to overlapping scopes).
  Hence the class of \lambdatg{s} over $\bsig$ is not closed under converse functional bisimulation.
%

  For showing (\ref{prop:lambdatgs:not:closed:under:funbisim:convfunbisim:item:iii}),
  consider the following term graphs over $\siglambdaij{1}{1}$:
  \graphs
  {
          \graph{\atgi{1}''}{lambdatgs_not_closed_under_funbisim_convfunbisim_item_iii_g1}
          \graph{\atgi{0}''}{lambdatgs_not_closed_under_funbisim_convfunbisim_item_iii_g0}
  }
  Then it holds that $\atgi{1}'' \funbisim \atgi{0}''$.         
  However,        
  while $\atgi{1}''$ admits a correct abstraction-prefix function,
  and hence is a \lambdatg\ over $\siglambdaij{1}{1}$,
  this is not the case for $\atgi{0}''$ (due to overlapping scopes).
  Therefore the class of \lambdatg{s} over $\siglambdaij{1}{1}$ is not closed under functional bisimulation.

  For showing (\ref{prop:lambdatgs:not:closed:under:funbisim:convfunbisim:item:iv}),
  consider the following term graphs over $\siglambdaij{1}{2}$: 
  \graphs{
          \graph{\atg'''_{1}}{lambdatgs_over_siglambda12_not_closed_under_funcbisim_g1}
          \graph{\atg'''_{0}}{lambdatgs_over_siglambda12_not_closed_under_funcbisim_g0}
          }%
  Then it holds that $\atg'''_{1} \funbisim \atg'''_{0}$.         
  However,        
  while $\atg'''_{1}$ admits a correct abstraction-prefix function,
  and hence is a \lambdatg\ over $\siglambdaij{1}{2}$,
  this is not the case for $\atg'''_{0}$ (overlapping scopes).
  Therefore the class of \lambdatg{s} over $\siglambdaij{1}{2}$ is not closed under functional bisimulation.
  The scopes defined implicitly by these graphs are larger than necessary: they
  do not exhibit `eager scope closure', see Section~\ref{sec:closed}.%
\end{proof}

As an easy consequence of Prop.~\ref{prop:lambdatgs:not:closed:under:bisim},
and of Prop.~\ref{prop:lambdatgs:not:closed:under:funbisim:convfunbisim},
 (\ref{prop:lambdatgs:not:closed:under:funbisim:convfunbisim:item:i})
 and (\ref{prop:lambdatgs:not:closed:under:funbisim:convfunbisim:item:ii}),
 together with the examples used in the proof, 
we obtain the following two propositions. 

\begin{proposition}\label{prop:not:closed:under:bisim:lambdahotgs:aplambdahotgs:siglambdai}
  Let $i\in\setexp{0,1}$. 
  None of the classes $\classlhotgsi{i}$ of \lambdahotg{s}, 
  or $\classlaphotgsi{i}$ of \lambdaaphotg{s} 
  are closed under bisimulations on the underlying term graphs. 
\end{proposition}

\begin{proposition}\label{prop:not:closed:under:funbisim:convfunbisim:lambdahotgs:aplambdahotgs:siglambdai}
  The following statements hold:
  \begin{enumerate}[(i)]\itemizeprefs
    \item{}\label{prop:not:closed:under:funbisim:convfunbisim:lambdahotgs:aplambdahotgs:siglambdai:item:i}
      Neither the class $\classlhotgsi{0}$ 
      nor the class $\classlaphotgsi{0}$ 
      is closed under functional bisimulations, or under converse functional bisimulations, on the underlying term graphs.
    \item{}\label{prop:not:closed:under:funbisim:convfunbisim:lambdahotgs:aplambdahotgs:siglambdai:item:ii} 
      Neither the class $\classlhotgsi{1}$ of \lambdahotg{s} 
      nor the class $\classlaphotgsi{1}$ of \lambdaaphotg{s} 
      is closed under converse functional bisimulations on underlying term graphs. 
  \end{enumerate}  
\end{proposition}
Note that Prop.~\ref{prop:not:closed:under:funbisim:convfunbisim:lambdahotgs:aplambdahotgs:siglambdai},
          (\ref{prop:not:closed:under:funbisim:convfunbisim:lambdahotgs:aplambdahotgs:siglambdai:item:i})
is a strengthening of the statement of Prop.~\ref{prop:no:extension:funbisim:siglambda1} earlier.

\section{Closed under functional bisimulation}
  \label{sec:closed}

The negative results gathered in the last section might seem to show our enterprise in a quite poor state: 
For the classes of \lambdatg{s} we introduced, 
Prop.~\ref{prop:lambdatgs:not:closed:under:funbisim:convfunbisim}
only leaves open the possibility that the class $\classltgsi{1}$ is closed under functional bisimulation.
Actually, $\classltgsi{1}$ is closed (we do not prove this here), 
but that does not help us any further, because the correspondences in
Thm.~\ref{thm:corr:laphotgs:ltgs} do not apply to this class, and worse still,
Prop.~\ref{prop:forgetful} rules out simple correspondences for $\classltgsi{1}$.
So in this case we are left  
without the satisfying correspondences to \lambdahotg{s} and \lambdaaphotg{s}
that yet exist for the other classes of \lambdatg{s}, but which in their turn are not closed under functional bisimulation. 

But in this section we establish that the class $\classltgsij{1}{2}$ is very useful after all: its restriction 
to term graphs with eager application of scope closure is in fact closed under functional bisimulation. 

The reason for the non-closedness of $\classltgsij{1}{2}$ under functional bisimulation consists in the fact 
that \lambdatg{s} over $\siglambdaij{1}{2}$ do not necessarily exhibit `eager scope closure': 
for example in the term graph $\atgi{1}'''$ 
from the proof of Prop.~\ref{prop:lambdatgs:not:closed:under:funbisim:convfunbisim}, (\ref{prop:lambdatgs:not:closed:under:funbisim:convfunbisim:item:iv}),
the scopes of the two topmost abstractions are not closed on the paths to variable occurrences belonging to the bottommost abstractions.
For the following variation $\tilde{\atg}_{1}$ of $\atgi{1}$ with eager scope closure
the problem disappears:
\\
  \graphs{
          \graph{\tilde{\atg}_{1}}{lambdatgs_over_siglambda12_not_closed_under_funcbisim_eager_g1}
          \graph{\tilde{\atg}_{0}}{lambdatgs_over_siglambda12_not_closed_under_funcbisim_eager_g0}
          }
\\
Its bisimulation collapse $\tilde{\atg}_{0}$ has again a correct \absprefix\ function and hence is a \lambdatg.

\begin{definition}[eager-scope, and fully back-linked, \lambdatg{s}]%
    \label{def:eager:scope:fully:back-linked:lambdatg}\normalfont
  Let $\atg = \tuple{\verts,\svlab,\svargs,\sroot}$ be a \lambdatg\ over $\siglambdaij{1}{j}$ for $j\in\setexp{1,2}$ 
  with \absprefix\ function~$\sabspre \funin \verts \to \verts^*$.
  We call $\atg$ an \emph{eager-scope \lambdatg\ (over $\siglambdaij{1}{j}$)} if it holds:
  \begin{equation*}
    \begin{aligned}
    \forallstzero{\bvert,\avert\in\verts}\, 
      \forallstzero{\apre\in\verts^*\!}
                    \abspre{\bvert} = \stringcon{\apre}{\avert}
                      \;\;\Rightarrow\;\;
                      & 
                      \exists{n\in\nats\,}\existsstzero{\bverti{1},\ldots,\bverti{n-1}\in\verts} 
                      \\[-0.25ex]
                      & \hspace*{6ex}
                        \bvert \tgsucc \bverti{1} \tgsucc \ldots \tgsucc \bverti{n-1} \tgsucci{0} \avert
                      \\[-0.25ex]
                      & \hspace*{6ex}
                          \;\:\slogand\:\;
                        \bverti{n-1} \in\vertsof{\snlvar}
                          \;\:\slogand\:\; 
                        \forallst{1\le i\le n-1}
                                 {\abspre{\bvert} \le \abspre{\bverti{i}}} \punc{,}
    \end{aligned}          
  \end{equation*}
  that is, if for every vertex $\bvert$ in $\atg$ with a non-empty abstraction-prefix $\abspre{\bvert}$ that ends with $\avert$
  there exists a path from $\bvert$ to $\avert$ in $\atg$ via vertices with abstraction-prefixes that extend $\abspre{\bvert}$
  and finally a variable-occurrence vertex before reaching $\avert$. 
  By $\classeagltgsij{i}{j}$ we denote the subclass of $\classltgsij{i}{j}$ consisting of all
  eager-scope-\lambdatg{s}. 
  And we say that $\atg$ is \emph{fully back-linked} if it holds:
  \begin{equation}\label{eq:def:fully:back-linked:lambdatg}
    \forallstzero{\bvert,\avert\in\vertsi{1}}\, 
      \forallst{\apre\in\vertsi{1}^*\!}
               {\,
                \abspre{\bvert} = \stringcon{\apre}{\avert}
                  \;\;\Rightarrow\;\;
                \bvert \pathto \avert} \punc{,}
  \end{equation}
  that is, if for all vertices $\bvert$ of $\atgi{1}$, the last vertex $\avert$ in the \absprefix\ of $\bvert$ is reachable from $\avert$.
  Note that eager-scope implies fully back-linkedness for \lambdatg{s}.
\end{definition}



\begin{lemma}\label{lem:preserve:lambdatgs}
  Let $\atg$ 
  be a fully back-linked \lambdatg\ 
                                    in $\classltgsij{1}{2}$
  with vertex set $\verts$, and let $\sabspre$ be its \absprefix\ function. 
  Let $\atgacc$ be a term graph over $\siglambdaij{1}{2}$ (thus in $\classtgssiglambdaij{1}{2}$) 
  such that $\atg \funbisimi{\sahom} \atgacc$. 
  Then it holds:
  \begin{equation}\label{eq:lem:preserve:lambdatgs}
    \forallst{\averti{1},\averti{2}\in\verts}
             {\; 
              \ahom{\averti{1}} = \ahom{\averti{2}}
                \;\;\Rightarrow\;\;
              \ahomext{\abspre{\averti{1}}} = \ahomext{\abspre{\averti{2}}}} 
  \end{equation}
  where $\sahomext$ is the homomorphic extension of $\sahom$ to words over $\verts$.
\end{lemma}

\begin{proofidea}
  Let $\atgi{1}$, $\atgi{2}$ be as assumed in the lemma,
  and let $\sahom$ be a homomorphism that witnesses $\atgi{1} \funbisimi{\sahom} \atgi{2}$. 
  We will use the following distance parameter for vertices of $\atgi{1}$:
  Let, for all $\bvert\in\vertsi{1}$, 
  $\distabs{\sabspre}{\bvert}$
  be either 0 if $\abspre{\bvert}$ is empty,
  or otherwise 
  the minimum length of a path in $\atgi{1}$
  from $\bvert$ to the last vertex in the \absprefix\ $\abspre{\bvert}$.
  Thus due to \eqref{eq:def:fully:back-linked:lambdatg}, $\distabs{\sabspre}{\bvert}\in\nats$ for all vertices $\bvert$ of $\atgi{1}$.
  Now \eqref{eq:lem:preserve:lambdatgs} can be proved
  by induction on $\max\setexp{\distabs{\sabspre}{\averti{1}},\distabs{\sabspre}{\averti{2}}}$
  with a subinduction on $\max\setexp{\length{\abspre{\averti{1}}},\length{\abspre{\averti{2}}}}$. 
\end{proofidea}

This lemma is the crucial stepping stone for the proof of the following theorem.

\begin{theorem}[preservation of \lambdatg{s} over $\siglambdaij{1}{2}$ under homomorphism]%
    \label{thm:preserve:lambdatgs}
  Let $\atg$ and $\atgacc$ be term graphs over $\siglambdaij{1}{2}$
  such that $\atg$ is a \lambdatg\ in $\classltgsij{1}{2}$, 
  and $\atg \funbisimi{\sahom} \atgacc$ holds for a homomorphism $\sahom$. 
 
  If $\atg$ is fully back-linked, then also $\atgacc$ is a \lambdatg\ in $\classltgsij{1}{2}$, 
  which is fully back-linked. If, in addition, $\atg$ is an eager-scope \lambdatg, then so is $\atgacc$.  
\end{theorem}

\begin{corollary}\label{cor:closed:under:fun:bisim}
  The subclass $\classeagltgsij{1}{2}$ of the class $\classtgssiglambdaij{1}{2}$ 
  that consists of all eager-scope \lambdatg{s} in $\classltgsij{1}{2}$
  is closed under functional bisimulation. 
\end{corollary}

Since the counterexample in the proof of 
Prop.~\ref{prop:lambdatgs:not:closed:under:funbisim:convfunbisim},
      (\ref{prop:lambdatgs:not:closed:under:funbisim:convfunbisim:item:iii}) 
used eager-scope \lambdatg{s}, it
rules out a statement analogous to Cor.~\ref{cor:closed:under:fun:bisim} for the class $\classltgsij{1}{1}$.
Such a statement for $\classltgsij{0}{1}$ and $\classltgsij{0}{2}$ is ruled out similarly, 
with respect to an appropriate definition of `eager-scope' for \lambdatg{s} over $\siglambdaij{0}{1}$ and $\siglambdaij{0}{2}$.%
  

\begin{corollary}
  Let $\sahom$ be a functional bisimulation from an eager-scope \lambdatg~$\altg$ over $\siglambdaij{1}{2}$
  to a term graph $\atgacc$ over $\siglambdaij{1}{2}$ ($\sahom$ witnesses $\altg \funbisimi{\sahom} \atgacc$).
  Then $\atgacc$ is an eager-scope \lambdatg\ as well,
  and $\sahom$ extends to a functional bisimulation
  from $\ltgstolaphotgsij{1}{2}{\altg}$ to $\ltgstolaphotgsij{1}{2}{\atgacc}$
  (thus $\sahom$ also witnesses $\ltgstolaphotgsij{1}{2}{\altg} \funbisimi{\sahom} \ltgstolaphotgsij{1}{2}{\atgacc}$).
\end{corollary}

\section{Conclusion}
  \label{sec:conclusion}

We first defined higher-order term graph representations for cyclic \lambda-terms:
\begin{itemize}\itemizeprefs
  \item 
    \lambdahotg{s} in $\classlhotgsi{i}$,  
    an adaptation of Blom's `higher-order term graphs'~\cite{blom:2001},
    which possess a scope function that maps every abstraction vertex $\avert$ to the set of vertices that are in the scope of $\avert$. 
  \item 
    \lambdaaphotg{s} in $\classlaphotgsi{i}$,
    which instead of a scope function carry an \absprefix\ function 
    that assigns to every vertex $\bvert$ information about the scoping structure relevant for~$\bvert$. 
    Abstraction prefixes are closely related to the notion of `generated subterms' for \lambdaterms\ \cite{grab:roch:2012}.
    The correctness conditions here are simpler and more intuitive than for \lambdahotg{s}. 
\end{itemize}
These classes are defined for $i\in\setexp{0,1}$, according to whether variable occurrences have back-links
to abstractions (for $i=1$) or not (for $i=0$). 
Our main statements about these classes are:
\begin{itemize}\itemizeprefs
  \item 
    a bijective correspondence between $\classlhotgsi{i}$ and $\classlaphotgsi{i}$ 
    via mappings
    $\slhotgstolaphotgsi{i}$
    and
    $\slaphotgstolhotgsi{i}$
    that preserve and reflect the sharing order (Thm.~\ref{thm:corr:lhotgs:laphotgs});
 \item 
   the naive approach to implementing homomorphisms on theses classes 
   (ignoring all scoping information and using only the underlying first-order term graphs) 
   fails (Prop.~\ref{prop:forgetful}).
\end{itemize}
The latter was the motivation to consider first-order term graph implementations with scope delimiters:
\begin{itemize}\itemizeprefs
  \item 
    \lambdatg{s} in $\classltgsij{i}{j}$ 
    (with $i\in\setexp{0,1}$ and $j=2$ or $j=1$ for scope delimiter vertices with or without back-links, respectively),
    which are first-order term graphs without a higher-order concept, but for which correctness conditions are formulated via
    the existence of an \absprefix\ function. 
\end{itemize}
The most important results linking these classes with \lambdaaphotg{s} are:
\begin{itemize}\itemizeprefs
  \item 
    an `almost bijective' correspondence
    between the classes  $\classlaphotgsi{i}$ 
    and $\classltgsij{i}{j}$ 
    via mappings
    $\slaphotgstoltgsij{i}{j}$
    and
    $\sltgstolaphotgsij{i}{j}$
    that preserve and reflect the sharing order
    (Thm.~\ref{thm:corr:laphotgs:ltgs});
  \item 
    the subclass $\classeagltgsij{1}{2}$ of eager-scope \lambdatg{s} in $\classltgsij{1}{2}$
    is closed under homomorphism (Cor.~\ref{cor:closed:under:fun:bisim}).
\end{itemize}
The correspondences together with the closedness result allow us to derive
methods to handle homomorphisms between eager higher-order term graphs in $\classlhotgsi{1}$ and $\classlaphotgsi{1}$
in a straightforward manner by implementing 
them via homomorphisms between first-order term graphs in $\classltgsij{1}{2}$. 


\hspace{1cm}
\begin{tikzpicture}[>=stealth]
\matrix[row sep=0.8cm,column sep=-4.1mm]{
\node(tl){\phantom{I}};&&
\node(t){$\classeaglhotgsi1$};&&
\node(tr){\phantom{I}};\\
\node(ml){\phantom{I}};&&
\node(m){$\classeaglaphotgsi1$};&&
\node(mr){\phantom{I}};\\
\node(bl){\phantom{I}};&&
\node(b){$\classeagltgsij12$};&&
\node(br){\phantom{I}};\\
};
\draw[->](tl) to node[left]{$\slhotgstolaphotgsi1$}  (ml);
\draw[->](ml) to node[left]{$\slaphotgstoltgsij12$}  (bl);
\draw[->](br) to node[right]{$\sltgstolaphotgsij12$} (mr);
\draw[->](mr) to node[right]{$\slaphotgstolhotgsi1$} (tr);
\end{tikzpicture}
\hspace{2cm}
\begin{tikzpicture}[>=stealth]
\matrix[row sep=0.8cm,column sep=0.8cm]{
\node(tl){$\alhotgi0$};&
\node(tr){$\alhotgi0'$};\\
\node(ml){$\alaphotgi1$};&
\node(mr){$\alaphotgi1'$};\\
\node(bl){$\atg$};&
\node(br){$\atg'$};\\
};
\draw[|->](tl) to node[below, at end]{\scriptsize $h$}(tr);
\draw[|->](ml) to node[below, at end]{\scriptsize $h$}(mr);
\draw[|->](bl) to node[below, at end]{\scriptsize $h'$}(br);
\draw[|->](tl) to node[left]{$\slhotgstolaphotgsi1$} (ml);
\draw[|->](ml) to node[left]{$\slaphotgstoltgsij12$} (bl);
\draw[|->](br) to node[right]{$\sltgstolaphotgsij12$} (mr);
\draw[|->](mr) to node[right]{$\slaphotgstolhotgsi1$} (tr);
\end{tikzpicture}
\\
For example, the property that a unique maximally shared form exists 
for \lambdatg{s} in $\classltgsij{1}{2}$ 
(which can be computed as the bisimulation collapse that is guaranteed to exist for first-order term graphs)
can now be transferred to eager-scope \lambdaaphotg{s} and \lambdahotg{s}
via the correspondence mappings (see the diagram above). 
For this to hold it is crucial that $\classeagltgsij{1}{2}$ is closed under homomorphism,
and that the correspondence mappings preserve and reflect the sharing order.  
The maximally shared form  $\max_{\classeaglhotgsi1}(\alhotg)$ of an eager \lambdahotg\ $\alhotg$ can furthermore
be computed as:
\begin{center}
  $\max_{\classeaglhotgsi1}(\alhotg) =
   (\scompfuns{\slaphotgstolhotgsi1}{\scompfuns{\sltgstolaphotgsij12}{\scompfuns{\max_{\classeagltgsij12}}{\scompfuns{\slaphotgstoltgsij12}{\slhotgstolaphotgsi1}}}})(\alhotg)$.
\end{center}
where $\max_{\classeagltgsij{1}{2}}$ maps every \lambdatg\ in $\classltgsij{1}{2}$ to its bisimulation collapse. 
For obtaining $\max_{\classeagltgsij{1}{2}}$ fast algorithms for computing
the bisimulation collapse of first-order term graphs can be utilized.

While we have explained this result here only for term graphs with eager
scope-closure, the approach can be generalized to non-eager-scope term graphs.
To this end scope delimiters have to be placed also underneath variable
vertices. Then variable occurrences do not implicitly close all open extended
scopes, but every extended scope that is open at some position must be
closed explicitly by scope delimiters on all (maximal) paths from that position.
The resulting graphs are fully back-linked, and then Thm.~\ref{thm:preserve:lambdatgs}
guarantees that the arising class of \lambdatg{s} is again closed under homomorphism.

For our original intent of getting a grip on maximal subterm sharing in the
\lambda-calculus with \stxtletrec\ or \mu, however, only eager scope-closure is
practically relevant, since it facilitates a higher degree of sharing.

Ultimately we expect that these results allow us to develop solid formalizations
and methods for subterm sharing in higher order languages with sharing
constructs. 

\vspace*{-2ex}
\paragraph{Acknowledgement.}
 We want to thank the reviewers for their helpful comments, and for pointing out 
 a number of inaccurate details in the submission that we have remedied for obtaining this version.  

\bibliography{ltgs}  

\end{document}

%% file: defs.tex
\newcommand{\inMath}[1]{\ensuremath{#1}}

\newcommand\vcentered[1]{\raisebox{-0.5\height}{#1}}

\newcommand\graphs[1]{\par\vspace{2mm}\hfill{#1}\ \vspace{2mm}\par\noindent}
\newcommand\graph[2]{\vcentered{$#1$:~~} \vcentered{\fig{#2}}\hfill}

\newcommand\partialTo\rightharpoonup

\newcommand{\tuple}[1]{\langle #1 \rangle}
\newcommand\tuplespace{\hspace*{0.5pt}}
\newcommand\pair[2]{\tuple{#1, \tuplespace #2}}

\newcommand\fig[1]{\includegraphics[scale=0.84]{figs/{{#1}}}}
\newcommand{\lambdacalculus}{$\lambda$\nb-calculus}

\newcommand{\lambdaterm}{$\lambda$\nb-term}
\newcommand{\lambdaterms}{\lambdaterm{s}}

\let\oldlambda\lambda
\renewcommand\lambda{\inMath\oldlambda}
\let\oldalpha\alpha
\renewcommand\alpha{\inMath\oldalpha}
\let\oldmu\mu
\renewcommand\mu{\inMath\oldmu}

\newcommand{\stxtletrec}{\ensuremath{\text{\normalfont\sf letrec}}}

\newcommand{\sslabs}{\lambda}
\newcommand{\slabs}[1]{\sslabs{#1}}
\newcommand{\labs}[2]{\slabs{#1}.\,{#2}}
\newcommand{\sslapp}{@}
\newcommand{\slapp}{\hspace*{1.5pt}}
\newcommand{\lapp}[2]{{#1}\slapp{#2}}

\newcommand{\snlvar}{\mathsf{0}}

\newcommand{\snlvarsucc}{\mathsf{S}}

\newcommand{\ssin}{{\textbf{in}}}

\newcommand{\sletrec}{\textbf{letrec}}

\newcommand{\letrecin}[2]{\sletrec\;{#1}\;\ssin\;{#2}}

\newcommand{\arecvar}{f}
\newcommand{\brecvar}{g}

\makeatletter
\def\a#1{\reflectbox{$\m@th#1{\lambda}$}}
\def\adbmal{\inMath{\mathpalette{\a}{}}}
\makeatother
\newcommand{\avar}{x}
\newcommand{\bvar}{y}
\newcommand{\cvar}{z}
\newcommand{\dvar}{u}

\newcommand{\asig}{\Sigma}
\newcommand{\bsig}{\Delta}
\newcommand{\asigacc}{\Sigma'}

\newcommand{\sarity}{{\mit ar}}
\newcommand{\arity}{\funap{\sarity}}

\newcommand{\afunsym}{\mathsf{f}}

\newcommand{\CRS}{${\text{CRS}}$}

\newcommand{\apath}{\pi}

\newcommand\unfold\bigtriangledown
\newcommand{\funin}{\mathrel{:}}
\newcommand{\funap}[2]{{#1}({#2})}
\newcommand{\bfunap}[3]{{#1}({#2},\hspace*{0.5pt}{#3})}

\newcommand{\indap}[2]{#1_{#2}}
\newcommand{\sdefdby}{{:=}}
\newcommand{\defdby}{\mathrel{\sdefdby}}
\newcommand{\length}[1]{\left|{#1}\right|}
\newcommand{\nb}{\nobreakdash}

\newcommand{\sdom}{\textrm{dom}}
\newcommand{\dom}{\funap{\sdom}}

\newcommand{\simage}{\textrm{im}}
\newcommand{\image}{\funap{\simage}}

\newcommand{\sfuncomp}{\circ}
\newcommand{\scompfuns}[2]{{#1}\mathrel{\sfuncomp}{#2}}
\newcommand{\compfuns}[2]{\funap{\scompfuns{#1}{#2}}}

\newcommand{\sproji}{\indap{\pi}}

\newcommand{\eqcl}[2]{\indap{[{#1}]}{#2}}

\newcommand{\existsst}[2]{\exists{#1}.\;{#2}}
\newcommand{\forallst}[2]{\forall{#1}.\;{#2}}
\newcommand{\existsstzero}[1]{\exists{#1}.\;}
\newcommand{\forallstzero}[1]{\forall{#1}.\;}

\newcommand{\emptyword}{\epsilon}
\newcommand{\sstringcon}{\hspace*{1pt}}
\newcommand{\stringcon}[2]{{#1}{\sstringcon}{#2}}

\newcommand{\mcdots}{\mathrel{\cdots}}

\newcommand{\srestrictfunto}[2]{{#1}{\mid}_{#2}}
\newcommand{\restrictfunto}[2]{\funap{\srestrictfunto}}

\newcommand{\sidfun}{\text{\normalfont id}}

\newcommand{\sidfunon}{\indap{\sidfun}}

\newcommand{\slogand}{{\wedge}}
\newcommand{\logand}{\mathrel{\slogand}}
\newcommand{\slogor}{{\vee}}
\newcommand{\logor}{\mathrel{\slogor}}
\newcommand{\ssbinrelcomp}{\cdot}
\newcommand{\sbinrelcomp}[2]{{#1}\mathrel{\ssbinrelcomp}{#2}}
\newcommand{\binrelcomp}[2]{\mathrel{\sbinrelcomp{#1}{#2}}}
\newcommand{\subap}[2]{{#1}_{#2}}
\newcommand{\supap}[2]{{#1}^{#2}}
\newcommand{\bpap}[3]{{#1}_{#2}^{#3}}
\newcommand{\pbap}[3]{{#1}^{#2}_{#3}}
\newcommand{\descsetexpmid}{\mathrel{\vert}}

\newcommand{\descsetexpBigmid}{\mathrel{\Big\vert}}

\newcommand{\descsetexp}[2]{\left\{{#1}\descsetexpmid{#2}\right\}}

\newcommand{\descsetexpBig}[2]{\Bigl\{{#1}\descsetexpBigmid{#2}\Bigr\}}

\newcommand{\setexp}[1]{\left\{{#1}\right\}}

\newcommand{\spowersetof}{\powerset}
\newcommand{\powersetof}{\funap{\spowersetof}}
\newcommand{\punc}[1]{\inMath{\hspace*{3pt}{#1}}}
\newcommand{\nats}{\mathbb{N}}

\definecolor{azure}{rgb}{0.94,1.00,1.00}
\definecolor{blue}{rgb}{0,0,0.5}
\definecolor{brown}{rgb}{.75,.25,.25}
\definecolor{cyan}{rgb}{0.25,0.88,0.82}
\definecolor{chocolate}{rgb}{0.82,0.41,0.12}
\definecolor{darkcyan}{rgb}{0.5,0,1}
\definecolor{darkgreen}{rgb}{0,0.39,0}
\definecolor{darkmagenta}{rgb}{0.5,0,0.5}
\definecolor{firebrick}{RGB}{175,25,25}
\definecolor{forestgreen}{rgb}{0.13,0.55,0.13}
\definecolor{lightcyan}{rgb}{0.88,1.00,1.00}
\definecolor{lightpink}{rgb}{1.00,0.71,0.76}
\definecolor{lightyellow}{rgb}{1.00,1.00,0.88}
\definecolor{lightgoldenrod}{rgb}{0.83,0.97,0.51}
\definecolor{lightgoldenrodyellow}{rgb}{0.98,0.98,0.82}
\definecolor{lightskyblue}{rgb}{0.53,0.81,0.98}
\definecolor{moccasin}{rgb}{1.00,0.89,0.71}
\definecolor{magenta}{rgb}{1,0,1}
\definecolor{navyblue}{rgb}{0,0,0.5}
\definecolor{orange}{rgb}{1.0,0.65,0.0}
\definecolor{orangered}{rgb}{1.0,0.27,0.0}
\definecolor{palegreen}{rgb}{0.60,0.98,0.60}
\definecolor{powderblue}{rgb}{0.69,0.88,0.90}
\definecolor{purple}{rgb}{1,0.5,1}
\definecolor{royalblue}{RGB}{65,105,225}
\definecolor{mediumblue}{RGB}{0,0,205}
\definecolor{cornflowerblue}{RGB}{100,149,237}
\definecolor{springgreen}{rgb}{0.0,1.0,0.5}
\definecolor{turquoise}{rgb}{0.25,0.88,0.82}
\definecolor{snow}{rgb}{1.00,0.98,0.98}
\definecolor{tan}{rgb}{0.82,0.71,0.55}
\definecolor{red}{rgb}{1,0,0}

\newcommand{\safun}{f}

\newcommand{\atg}{G}
\newcommand{\atgi}{\indap{\atg}}
\newcommand{\atgacc}{G'}

\newcommand{\verts}{V}
\newcommand{\vertsof}{\funap{\verts\!}}
\newcommand{\svlab}{\mathit{lab}}
\newcommand{\vlab}{\funap{\svlab}}
\newcommand{\svargs}{\mathit{args}}
\newcommand{\vargs}{\funap{\svargs}}
\newcommand{\sroot}{\mathit{r}}
\newcommand{\vertsacc}{V'}

\newcommand{\svlabacc}{\mathit{lab}'}
\newcommand{\vlabacc}{\funap{\svlabacc}}
\newcommand{\svargsacc}{\mathit{args}'}

\newcommand{\srootacc}{\mathit{r}'}
\newcommand{\vertsi}{\indap{V}}
\newcommand{\vertsiof}[1]{\funap{\vertsi{#1}\!}}
\newcommand{\svlabi}{\indap{\mathit{lab}}}
\newcommand{\vlabi}[1]{\funap{\svlabi{#1}}}
\newcommand{\svargsi}{\indap{\mathit{args}}}
\newcommand{\vargsi}[1]{\funap{\svargsi{#1}}}
\newcommand{\srooti}{\indap{\mathit{r}}}

\newcommand{\avert}{v}
\newcommand{\bvert}{w}

\newcommand{\averti}{\indap{\avert}}
\newcommand{\bverti}{\indap{\bvert}}

\newcommand{\bvertbp}{\bpap{\bvert}}

\newcommand{\avertacc}{v'}
\newcommand{\bvertacc}{w'}

\newcommand{\avertacci}{\indap{\avertacc}}

\newcommand{\siglambda}{\supap{\asig}{\lambda}}

\newcommand{\siglambdai}[1]{\pbap{\asig}{\sslabs}{{#1}}}
\newcommand{\siglambdaij}[2]{\pbap{\asig}{\sslabs}{{#1},{#2}}}

\newcommand{\stgsucc}{{\rightarrowtail}}
\newcommand{\tgsucc}{\mathrel{\stgsucc}}
\newcommand{\stgsucci}{\indap{\rightarrowtail}}
\newcommand{\tgsucci}[1]{\mathrel{\stgsucci{#1}}}

\newcommand{\stgsuccacci}{\indap{\rightarrowtail'}}
\newcommand{\tgsuccacci}[1]{\mathrel{\stgsuccacci{#1}}}

\newcommand{\stgsuccis}[2]{{{}^{#2}\hspace*{-1.5pt}\stgsucc_{#1}}}
\newcommand{\tgsuccis}[2]{\mathrel{\stgsuccis{#1}{#2}}}
\newcommand{\stgsuccisstar}[2]{(\stgsuccis{#1}{#2})^*}

\newcommand{\snodels}{{\#\text{\normalfont del}}}
\newcommand{\nodels}{\bfunap{\snodels}}

\newcommand{\spathto}{\rightarrowtail\hspace*{-5pt}^*}
\newcommand{\pathto}{\mathrel{\spathto}}

\newcommand{\sahom}{h}
\newcommand{\ahom}{\funap{\sahom}}
\newcommand{\sahomext}{\bar{\sahom}}
\newcommand{\ahomext}{\funap{\sahomext}}
\newcommand{\sahomextext}{\bar{\sahomext}}

\newcommand{\siso}{\sim}

\newcommand{\sbisim}{%
    \setbox0=\hbox{\kern-.1ex{$\leftrightarrow$}\kern-.1ex}
    \setbox1=\vbox{\hbox{\raise .1ex \box0}\hrule}%
    \inMath{\mathrel{\hbox{\kern.1ex\box1\kern.1ex}}}
  }
\newcommand{\bisim}{\mathrel{\sbisim}}
\newcommand{\bisimi}[1]{\mathrel{\sbisim_{#1}}}

\newcommand{\sbisimsubscript}{
    \setbox0=\hbox{\kern-.1ex{$\leftrightarrow$}\kern-.1ex}
    \setbox1=\vbox{\hbox{\raise .1ex \box0}\hrule}%
    \inMath{\mathrel{\hbox{\scalebox{0.75}{\box1}}}}
  }

\newcommand{\sfunbisim}{%
    \setbox0=\hbox{\kern-.1ex{$\rightarrow$}\kern-.1ex}
    \setbox1=\vbox{\hbox{\raise .1ex \box0}\hrule}%
    {\hbox{\kern.1ex\box1\kern.1ex}}
  }
\newcommand{\funbisim}{\mathrel{{\sfunbisim}}}

\newcommand{\sfunbisimi}{\indap{\sfunbisim}}
\newcommand{\funbisimi}[1]{\mathrel{\sfunbisimi{#1}}}

\newcommand{\sfunbisims}{\supap{\sfunbisim}}
\newcommand{\funbisims}[1]{\mathrel{\sfunbisims{#1}}}

\newcommand{\sfunbisimis}{\bpap{\sfunbisim}}
\newcommand{\funbisimis}[2]{\mathrel{\sfunbisimis{#1}{#2}}}

\newcommand{\sconvfunbisim}[1][]{%
    \setbox0=\hbox{\kern-.1ex{$\leftarrow$}\kern-.1ex}
    \setbox1=\vbox{\hbox{\raise .1ex \box0}\hrule}%
    \mathrel{\hbox{\kern.1ex\box1\kern.1ex}}
  }
\newcommand{\convfunbisim}{\mathrel{\sconvfunbisim}}

\newcommand{\sconvfunbisimi}{\indap{\sconvfunbisim}}
\newcommand{\convfunbisimi}[1]{\mathrel{\sconvfunbisimi{#1}}}

\newcommand{\sconvfunbisims}{\supap{\sconvfunbisim}}
\newcommand{\convfunbisims}[1]{\mathrel{\sconvfunbisims{#1}}}

\newcommand{\sconvfunbisimis}{\bpap{\sconvfunbisim}}
\newcommand{\convfunbisimis}[2]{\mathrel{\sconvfunbisimis{#1}{#2}}}

\newcommand{\sSfunbisim}{\sfunbisim^{\snlvarsucc}}
\newcommand{\Sfunbisim}{\mathrel{\sSfunbisim}}

\newcommand{\abisim}{R}

\newcommand{\sScope}{\mathit{Sc}}
\newcommand{\Scope}{\funap{\sScope}}
\newcommand{\sScopemin}{\mathit{Sc}^{-}}
\newcommand{\Scopemin}{\funap{\sScopemin}}
\newcommand{\sScopei}{\indap{\mathit{Sc}}}
\newcommand{\Scopei}[1]{\funap{\sScopei{#1}}}

\newcommand{\sbinders}{\mathrm{bds}}
\newcommand{\binders}{\funap{\sbinders}}

\newcommand{\sscopeforgetfullhotgsi}[1]{\sScope\mathit{F}_{#1}^{\sslabs}}
\newcommand{\scopeforgetfullhotgsi}[1]{\funap{\sscopeforgetfullhotgsi{#1}}}
\newcommand{\sabspreforgetfullaphotgsi}[1]{{\sabspre\mathit{F}_{#1}}\hspace*{-0.1em}^{(\sslabs)}}
\newcommand{\abspreforgetfullaphotgsi}[1]{\funap{\sabspreforgetfullaphotgsi{#1}}}

\newcommand{\lambdahotg}{$\lambda$\nb-ho-term-graph}
\newcommand{\alhotg}{{\cal G}}
\newcommand{\alhotgi}[1]{{\cal G}_{#1}}
\newcommand{\alhotgacc}{{\cal G}'}

\newcommand{\lambdaaphotg}{$\lambda$\nb-ap-ho-term-graph}
\newcommand{\alaphotg}{{\cal G}}
\newcommand{\alaphotgi}[1]{{\cal G}_{#1}}

\newcommand{\lambdatg}{$\lambda$\nb-term-graph}
\newcommand{\altg}{G}
\newcommand{\altgi}{\indap{\altg}}

\newcommand{\seag}{\text{\normalfont eag}}
\newcommand{\classlhotgsi}{\pbap{\cal H}{\sslabs}}
\newcommand{\classlaphotgsi}[1]{\subap{\cal H}{#1}{\hspace*{-0.1em}}^{(\sslabs)}}
\newcommand{\classeaglhotgsi}{{}^{\seag}\pbap{\cal H}{\sslabs}}
\newcommand{\classeaglaphotgsi}[1]{{}^{\seag}\subap{\cal H}{#1}{\hspace*{-0.1em}}^{(\sslabs)}}
\newcommand{\classltgsi}[1]{\subap{\cal T}{#1}{\hspace*{-0.2em}}^{(\sslabs)}}
\newcommand{\classltgsij}[2]{\subap{\cal T}{#1,#2}{\hspace*{-0.55em}}^{(\sslabs)}\!}

\newcommand{\classeagltgsij}[2]{{}^{\seag}\subap{\cal T}{#1,#2}{\hspace*{-0.55em}}^{(\sslabs)}\!}
\newcommand{\classtgssiglambdai}{\subap{\cal T}}
\newcommand{\classtgssiglambdaij}[2]{\subap{\cal T}{{#1},{#2}}}

\newcommand{\slhotgstolaphotgsi}{\indap{A}}
\newcommand{\lhotgstolaphotgsi}[1]{\funap{\slhotgstolaphotgsi{#1}}}
\newcommand{\slaphotgstolhotgsi}{\indap{B}}
\newcommand{\laphotgstolhotgsi}[1]{\funap{\slaphotgstolhotgsi{#1}}}

\newcommand{\slaphotgstoltgsij}[2]{G_{#1,#2}}
\newcommand{\laphotgstoltgsij}[2]{\funap{\slaphotgstoltgsij{#1}{#2}}}
\newcommand{\sltgstolaphotgsij}[2]{{\cal G}_{#1,#2}}
\newcommand{\ltgstolaphotgsij}[2]{\funap{\sltgstolaphotgsij{#1}{#2}}}

\newcommand{\sabspre}{P}
\newcommand{\abspre}{\funap{\sabspre}}
\newcommand{\sabsprei}{\indap{P}}
\newcommand{\absprei}[1]{\funap{\sabsprei{#1}}}
\newcommand{\absprefix}{ab\-strac\-tion-pre\-fix}

\newcommand{\apre}{p}

\newcommand{\sdistabs}[1]{d_{\lambda,{#1}}}
\newcommand{\distabs}[1]{\funap{\sdistabs{#1}}}

\newcommand{\prele}{\le}
\newcommand{\prelt}{<}

\newcommand{\aclass}{{\cal K}}

\newcommand{\classtgsover}{\funap{{\text{\normalfont TG}}}}